\documentclass[reqno]{amsart}
\usepackage[latin1]{inputenc}
\usepackage[TS1,T1]{fontenc}
\usepackage{color,ragged2e,newlfont,pifont,stmaryrd,xspace}
\usepackage{graphicx}
\usepackage{caption}
\usepackage{subcaption}
\usepackage{cases}
\usepackage{fancyhdr}
\newtheorem{theorem}{Theorem}[section]
\newtheorem{lemma}[theorem]{Lemma}
\newtheorem{proposition}[theorem]{Proposition}
\theoremstyle{definition}

\theoremstyle{remark}
\newtheorem{remark}[theorem]{Remark}

\numberwithin{equation}{section}

\begin{document}
\title{Analysis of an initial value problem for an
 extracellular and intracellular
 model of hepatitis C virus infection}


\author{Alexis NANGUE}
\address{Higher Teachers' Training College of the University of Maroua,
P.O.Box 55, Maroua, Cameroon} \curraddr{}
\email{alexnanga02@yahoo.fr}
\thanks{}

\author{Alan D. Rendall}
\address{Institute for Mathematics,
 Johannes Gutenberg University, Staudingerweg 9, 55099
Mainz, Germany} \curraddr{} \email{rendall@uni-mainz.de}
\thanks{}

\author{Brice  KAMMEGNE TCHEUGAM}
\address{African Institute for Mathematical Sciences, P.O. Box 608 Limbe, Cameroon}
\curraddr{} \email{brice.kammegne@aims-cameroon.org}
\thanks{}

\author{Patrick Steve  KAMDEM SIMO}
\address{Sub-regional Institute of Statistics and Applied Economics,
 P.O. Box  294 Yaound\'{e}, Cameroon}
\curraddr{} \email{ksimopatrick@gmail.com}
\thanks{}



\date{}

\dedicatory{}

\begin{abstract}
In this paper, a mathematical analysis of the global dynamics of a
viral infection model in vivo is carried out. We study the dynamics
of a hepatitis C virus (HCV) model, under therapy, that considers
both extracellular and intracellular levels of infection. At present
most mathematical modelling of viral kinetics after treatment only
addresses the process of infection of a cell by the virus and the
release of virions by the cell, while the processes taking place
inside the cell are not included. We prove that the solutions of the
new model with positive initial values are positive, exist globally
in time and are bounded. The model has two virus-free steady states.
They are distinguished by the fact that viral RNA is absent inside
the cells in the first state and present inside the cells in the
second. There are basic reproduction numbers associated to each of
these steady states. If the basic reproduction number of the first
steady state is less than one then that state is asymptotically
stable. If the basic reproduction number of the first steady state
is greater than one and that of the second less than one then the
second steady state is asymptotically stable. If both basic
reproduction numbers are greater than one then we obtain various
conclusions which depend on different restrictions on the parameters
of the model. Under increasingly strong assumptions we prove that
there is at least one positive steady state (infected equilibrium),
that there is a unique positive steady state and that the positive
steady state is stable. We also give a condition under which every
positive solution converges to a positive steady state. This is
proved by methods of Li and Muldowney. Finally we illustrate the
theoretical results by numerical simulations.
\end{abstract}
\maketitle

\section{Introduction}
Infection with hepatitis C virus (HCV) is one of the most common
causes of chronic liver disease. An account of its global
epidemiology can be found in \cite{Shepardetal2005}, where the
number of people infected worldwide is estimated as 123 million.
Achieving a sustained viral response (SVR), defined as undetectable
HCV-RNA in serum (viral load) 24 weeks after the end of treatment,
is the most effective way to prevent disease progression
\cite{SeeffandHoofnagle2003}. Recently the classical treatment
regimes with pegylated interferon (IFN) and ribavirin have been
improved on by the use of direct-acting antiviral agents (DAA). The
new treatments can produce a cure in more than 90\% of chronic cases
\cite{who2016}. In the past, mathematical models of the viral
dynamics of HCV have proven useful in describing the interaction
between the virus and host cells. In recent years, several papers on
the dynamics of HCV and other related pathogens such as the human
immunodeficiency virus (HIV) and the hepatitis B virus (HBV) have
appeared \cite{Anangue22019, Xiang, hattaf1, Gong, nangue2,
nangue3}. These studies have provided insights into viral
replication, cell death rate and treatment effectiveness but they
did not take into account the intracellular level of the infection.
\\\indent
In the basic model of virus dynamics, often used to describe the
dynamics of HCV, HBV and HIV infections, a simple view of viral
infection is proposed through the coupled evolution of three
populations: uninfected cells, infected cells and free virus
particles. The viral dynamics is therefore described by the temporal
evolution of the three populations. Mathematical modelling of HCV
infection and treatment has provided valuable insights into
viral-host-IFN dynamics \cite{Neumannetal1998} and has helped to
improve the treatment of HCV \cite{Herrmann2003}. In this model, the
virus is produced and released from productively infected cells into
the systemic circulation, where it can be cleared or infect further
target cells. It was shown using this model that the first phase of
viral decline is due to IFN acting to reduce the average rate of
virion production and release per infected cell, whereas the slower
second phase viral decline was attributed to the progressive loss of
infected cells \cite{Neumannetal1998}.
\\\indent
Denote by $T$, $I$ and $V$ the concentrations of healthy
hepatocytes, hepatocytes infected with HCV, and free HCV virions.
Because of the interpretation of these quantities they are
non-negative in any biologically relevant solution. The dynamics of
HCV infection is the result of the dynamics of the compartments $T$,
$I$, and $V$, and the various interactions between them. The
following system is a modification of an extracellular model given
in \cite{guedj2010}:
\begin{numcases}\strut
\nonumber \frac{d T}{dt} = s+r_{T}T\left(1-\frac{T+I}{T_\textrm{max}}\right) -d T - \frac{b TV}{T+I};\\
          \frac{d I}{dt} = r_{I} I\left(1-\frac{T+I}{T_\textrm{max}}\right) + \frac{b TV}{T+I}-\delta I;
          \label{cellmod}\\
\nonumber \frac{d V}{dt} = (1-\epsilon)p I -cV -\frac{b TV}{T+I}.
\end{numcases}

\noindent Its key features are as follows:

\noindent (i) The rate of change of the amount of healthy
hepatocytes $T$ is given by the first equation of (\ref{cellmod}).
Healthy hepatocytes are produced at a constant rate $s$ from an
external source and die at rate $d T$. The model in \cite{guedj2010}
has $s=0$. The population of uninfected hepatocytes is assumed to
maintain itself logistically, with homeostatic carrying capacity
$T_\textrm{max}$ as proposed in \cite{Anangue22019, Xiang}. Thus the
recruitment of healthy hepatocytes is given by
$r_{T}T\left(1-\frac{T+I}{T_\textrm{max}}\right)$, where $r_{T}$ is
the maximal per capita growth rate or the proliferation rate.
Virions infect the healthy hepatocytes at the rate $\frac{b
TV}{T+I}$, where $b$ is the rate of transmission of the infection,
an expression only defined when $T+I>0$. For the significance of
this term in the modelling of hepatitis we refer to \cite{hews2010}.
This standard incidence function replaces the mass action function
(used in the model of \cite{guedj2010}) which has been shown lead to
the unrealistic feature that a larger liver mass favours the
establishment of a chronic hepatitis infection. \\ \noindent (ii)
The second equation of (\ref{cellmod}) gives the rate of change of
infected cells $I$. The hepatocytes which are infected with HCV die
at rate $\delta$ per day so that $\frac{1}{\delta}$ is the life
expectancy of hepatocytes infected with HCV. Healthy hepatocytes
become infected at the rate $\frac{b TV}{T+I}$. As in
\cite{Anangue22019, Xiang}, we assume that hepatocytes infected with
HCV proliferate by a complete logistic term $r_{I}
I\left(1-\frac{T+I}{T_\textrm{max}}\right) $ , where $r_{I}$ is the
proliferation rate or maximal per capita growth rate of hepatocytes
infected with HCV. The model of \cite{guedj2010} has $r_I=0$.

\noindent (iii) The third equation of (\ref{cellmod}) gives the rate
of change of the free virus $V$. The infected hepatocytes produce
virus at rate $(1-\epsilon)pI$, and virus is cleared at the rate
$cV$. Also, the population of virions decreases due to the infection
at the rate $\frac{b TV}{T+I}$ : this is the absorption phenomenon
\cite{hattaf1}, which is not included in the model of
\cite{guedj2010}. The efficacy of treatment in blocking virion
production  is described by the parameter, $\epsilon$ whose value is
non-negative and less than one.

\indent The intracellular and cellular infection (ICCI) model is a
multi-scale model that encompasses the original cellular infection
(CI) model (\ref{cellmod}) but includes the viral production rate as
a dynamical process that may vary with time according to
intracellular treatment pressure and viral evolution. It is also a
modification of a model presented in \cite{guedj2010} and we adopt
some of the terminology of that reference in naming the models. To
avoid having too many parameters, the modelling of the intracellular
replication cycle is simplified to involve only the two
intracellular variables that are essential for RNA replication.
Hence the intracellular model is given by :
\begin{equation}\label{intramod}
 \left\{
  \begin{array}{ll}
\frac{d U}{dt} = \beta R\left(1-\frac{U}{U_\textrm{max}}\right)-\gamma U;     & \hbox{} \\
 \frac{d R}{dt} = \alpha U -\sigma R .    & \hbox{}
  \end{array}
\right.
\end{equation}
Here $R(t)$ is the number of positive genomic RNA strands that are
available for transcription and translation. It does not include RNA
which is packaged into the replication units which are responsible
for the production of new virus RNA in hepatitis C. $U(t)$ is the
number of RNA molecules  within the replication units which are
available as templates for RNA production with rate constant
$\alpha$. On the other hand the RNA included in $R(t)$ serves as a
template for the formation of replication units $U(t)$ with a
maximal rate constant $\beta$. This leads to a replication feedback
loop, where a large number of replication units can be formed and
function at the same time in each cell \cite{Daharietal2007}.
However, replication units are embedded in a replication complex,
including the vesicular membranous structure \cite{Eggeretal2002,
Gosertetal2003}, which requires a large amount of resources.
Limitations on these resources give rise to a maximum number of
replication units possible within a cell, $U_\textrm{max}$. Thus, we
assume here that the formation of replication units $U(t)$ is
rate-limited by $\beta R \left(1-\frac{U}{U_\textrm{max}}\right) $.
$R$ is lost by degradation with rate $\sigma$. In addition, we
assume that replication units $U$ are intrinsically unstable, and
that they are thus lost with a degradation rate constant $\gamma$.
Although there are no data in vivo for $\gamma$, there are good
indications in vitro that this rate is faster than the loss rate of
infected cells $(d)$ but slower than the viral clearance $(c)$
according to \cite{Daharietal2009}.
\\\indent
The link between the intracellular replication dynamics
(\ref{intramod}) and the cellular infection dynamics (\ref{cellmod})
is mediated by replacing the constant production/release rate, $p$,
in the cellular infection model with a time-dependent
production/release rate $p(t)= \rho R(t) $ \cite{guedj2010}, where
we assume that the packaged virus is exported on a rapid time scale.
Thus, in this paper we consider the full intracellular and cellular
infection (ICCI) model in presence of treatment, given by the
following system :
\begin{subequations}\label{a6}
\begin{align}
\frac{d T}{dt}&= s+r_{T}T\left(1-\frac{T+I}{T_\textrm{max}}\right) -d T - \frac{b TV}{T+I}; \label{fullICCImodel1} \\
  \frac{d I}{dt}&=r_{I} I\left(1-\frac{T+I}{T_\textrm{max}}\right) + \frac{b TV}{T+I}-\delta I ;  \label{fullICCImodel2} \\
  \frac{d V}{dt}&= \rho RI -cV -\frac{b TV}{T+I};  \label{fullICCImodel3} \\
 \frac{d U}{dt}&= \beta R\left(1-\frac{U}{U_\textrm{max}}\right)-\gamma U ;  \label{fullICCImodel4} \\
\frac{d R}{dt}&= \alpha(1-\epsilon) U -\sigma
R.\label{fullICCImodel5}
\end{align}
\end{subequations}
The system (\ref{a6}) is a modification of a system used in
\cite{guedj2010}. Note that the factor $1-\epsilon$ which represents
the effect of treatment occurs in a different place in (\ref{a6})
than in (\ref{cellmod}). This implements the fact that, as discussed
in \cite{guedj2010}, the primary effect of the DAA is to block the
synthesis of RNA.

The initial conditions associated to system (\ref{a6}) are given by
:
\begin{eqnarray}\label{a7}
&&T(0)=T_{0},\; I(0)=I_{0}, \; V(0)=V_{0} ,\;U(0)=U_{0}, \;
R(0)=R_{0}.
\end{eqnarray}
where the constants $T_0$, $I_0$, $V_0$, $U_0$ and $R_0$ are
positive. For biological significance of the parameters, four
assumptions are employed. (a) Due to the burden of supporting virus
replication, infected cells proliferate more slowly than uninfected
cells, i.e. $ r_{I} \leq r_{T}$. (b) To have a physiologically
realistic model, in an uninfected liver when $T_\textrm{max}$ is
reached, liver size should no longer increase, i.e. $s \leq d
T_\textrm{max} $. (c) Infected cells have a higher turnover rate
than uninfected cells, i.e. $ \delta \geq d$. (d) The rapid first
phase of viral decline is limited either by $c$ or by $\sigma$. If
$c < \sigma $, then the initial slope of decline is mainly due to
the clearance of virus ; if $\sigma < c $, then this is due to the
loss rate of genomic RNA and the export. If we assume that the first
phase of viral decline is due to the viral clearance
\cite{Daharietal2009, Ramratnametal1999}, then $\sigma
> c $. Hence the parameters are such that $d\le\delta < \gamma < c <
\sigma $.

\indent This paper is organized as follows. The positivity, global
existence and boundedness of solutions are obtained in
Section~\ref{sec2}. The equilibria of system (\ref{a6}) are studied
in Section~\ref{sec3} and the basic reproduction numbers of the
virus-free equilibria are given. The local asymptotic stability of
the virus-free steady states is established under appropriate
conditions. It is shown that when both reproductive numbers are
greater than one and $r_I>\delta$ there exists at least one infected
steady state. We identify conditions on the parameters under which
the infected steady state is unique and further conditions under
which it is locally asymptotically stable. Section~\ref{sec4} proves
a result on the global asymptotic stability of the virus-free steady
state $E_0$. In Section~\ref{sec5} a condition on the parameters is
identified under which every positive solution converges to a
positive steady state. In section~\ref{sec6}, numerical simulations
are carried out to illustrate the theoretical results obtained.
Finally, a brief discussion concludes the paper.
\section{Existence and global boundedness of solutions of the initial value problem (\ref{a6})-(\ref{a7})}\label{sec2}
\subsection{Existence, uniqueness and positivity of local and global solutions
of the initial value problem (\ref{a6}), (\ref{a7})} The main task
of this subsection is twofold. Firstly we are going to show that the
solution cannot approach the boundary of the domain of definition of
the system (\ref{a6}) arbitrarily closely and from this we deduce
the positivity. Secondly we show that the solution of the initial
value value problem (\ref{a6}), (\ref{a7}) is bounded on each finite
time interval. It is well known by the fundamental theory of
ordinary differential equations (ODE), that the system (\ref{a6})
has a unique local solution $(T(t),I(t),V(t),U(t),R(t))$ satisfying
the initial conditions (\ref{a7}) since the right hand side of the
system (\ref{a6}) is locally Lipschitz on the region where all the
variables are positive.
\begin{lemma}\label{lem11}
The infimum of $T(t) + I(t)$ is different from zero for any positive
solution on an interval $[0, t_{0})$, where $t_0=\infty$ is allowed.
\end{lemma}
\begin{proof}
The evolution equation of $T+I$ is given by :
\begin{equation}\label{evolt+i}
    \frac{d}{dt}(T+I)=s+(r_{T}T+r_{I}I)
    \left(1-\frac{T+I}{T_\textrm{max}}\right)-dT-\delta I.
\end{equation}
If the infimum of $T(t) + I(t)$ were zero then there would have to
exist a time $t_{1} \in [0, t_{0})$ such that $T(t_{1}) +
I(t_{1})<T_\textrm{max}$. Let $t_{2}$ be the infimum of the times
$t$ for which $T(t) + I(t)<T_\textrm{max}$ on $(t, t_{1})$. We do
not know a priori if $t_{2}=0$ or $t_{2}> 0$. Let $a = \max\{d,
\delta\} $. Then from (\ref{evolt+i}) on $(t_{2}, t_{1})$ we have
\begin{equation*}
\frac{d}{dt}(T+I) \geq s-a (T+I).
\end{equation*}
Hence
\begin{equation*}
(T + I)(t_{1}) \geq (T +
I)(t_{2})e^{-a(t_{1}-t_{2})}+sa^{-1}(1-e^{-a(t_{1}-t_{2})}).
\end{equation*}
\begin{description}
    \item If $t_{2}=0$ then  $(T + I)(t_{2})=I(0)+T(0)>0$.
    \item If $t_{2}>0$ then $(T + I)(t_{2})=T_\textrm{max} >0$.
\end{description}
Thus in both cases we get a positive lower bound for $(T +
I)(t_{2})$. On the other hand $e^{-a(t_{1}-t_{2})} \geq e^{-at_{1}}
$. This contradicts the assumption that the infimum of $T+I$ was
zero and completes the proof.
\end{proof}
\begin{remark}
As a consequence the solution cannot approach those points of the
boundary of the domain of definition of system (\ref{a6}) where the
right hand side of the equations does not have a continuous
extension and therefore if a solution exists on an interval
$[0,t_*)$
 it satisfies a bound of the form $T+I\geq C>0$.
\end{remark}
We show in the following  proposition that solutions of the initial
value problem (\ref{a6})-(\ref{a7}) are positive which means that
the model is well-posed biologically.
\begin{proposition}
Let $(T(t), I(t), V(t), U(t), R(t))$ be a solution of the initial
value problem (\ref{a6})-(\ref{a7}) on an interval $[0, t_{1})$ with
$t_{1}< +\infty$. If $T_0
> 0$, $I_0 > 0$, $V_0 > 0$, $U_0 > 0$ and $R_0 > 0$ then
$$  \liminf\limits_{t \to t_{0}}\min\{ T(t),\, I(t),\,
 V(t),\, U(t),\, R(t)\} > 0 .$$
\end{proposition}
\begin{proof}
For convenience we introduce the notation $X_1 = T$, $ X_2 = I$,
$X_3 = V$, $X_4 = U$, $X_5 = R$. Let $t^*$ be the supremum of times
$t$ for which $X_i(t) > 0$ on $[0,t)$ for all $i \in \{1 , 2 , 3 , 4
, 5\}$. Each $X_{i}$ satisfies an ordinary differential equation of
the form
$$\dot{X}_{i}=-X_{i}f_{i}(X)+g_{i}(X),$$ where $g_{i}$
are some functions of $(T, I, V, U, R)$ and $g_{i} \geq 0 $ for all
$ i \in \{1, 2, 3, 4, 5\} $. As consequence $\dot{X}_{i} \geq
-X_{i}f_{i}(X) $ and $ \frac{d}{dt}(\log X_{i}) \geq -f_{i}(X)$ on
$[0,t^*)$. Suppose that $t^*<t_1$. Then according to
Lemma~\ref{lem11},  $(T+I)^{-1}$ is known to be bounded for this
solution and therefore $f_{i}(T(t), I(t), V(t), U(t), R(t))$ is
bounded by a constant $M$. Hence $\frac{d}{dt}(\log X_{i}) \geq -M$
and $X_{i}(t)\geq X_{i}(0)e^{-M t^{*}} > 0 $. It follows that the
infimum of $X_{i}$ is strictly positive, contradicting the
assumption that $t^*<t_1$. Hence $t=t^*$ and this completes the
proof of the proposition.
\end{proof}
It will now be shown that all solutions of (\ref{a6})-(\ref{a7})
with positive initial data exist globally in time in the future.
\begin{theorem}\label{theo11}
The initial value problem (\ref{a6})-(\ref{a7})admits a unique
global solution defined on $[0, +\infty[$.
\end{theorem}
\begin{proof}
Taking the sum of equations (\ref{fullICCImodel1}) and
(\ref{fullICCImodel2}) shows that $$\frac{d(T+I)}{dt} \leq
s+\lambda(T+I), $$ where $\lambda=\max\{r_{I},r_{T}\}$. It follows
from this differential inequality that $T$ and $I$ are bounded on
any finite interval. Morover, taking the sum of equations
(\ref{fullICCImodel3}) and (\ref{fullICCImodel4}) shows that
$$\frac{d(U+R)}{dt} \leq \kappa(U+R), $$ where
$\kappa=\max\{\alpha,\beta\}$, and hence $R$ and $U$ are bounded on
any finite interval. Equation (\ref{fullICCImodel4}) implies that
$$ \frac{d V}{dt} \leq
\rho R I -c V, $$ which shows that $V(t)$ cannot grow faster than
linearly and is also bounded on any finite interval. By the
arguments above the solution on a finite maximum interval of
existence is positive and admits a positive lower bound for $T+I$.
By the estimates just proved it is bounded. Hence it remains in a
compact subset of the domain of definition of the system. The
standard continuation criterion for ODE then implies global
existence and this completes the proof of the theorem.
\end{proof}
\subsection{Global boundedness of
the solutions  for the initial value problem
 (\ref{a6})-(\ref{a7})}

We are now going to prove that the solution is globally bounded.

\begin{theorem}\label{theoglobound}
For any positive solution $(T, I, V, U, R)$ of the initial value
problem (\ref{a6}), (\ref{a7}) and any $\zeta>0$, we have that for
$t$ sufficiently large:
\begin{equation*}
    T(t)+I(t) \leq (1+\zeta)p_0, \; V(t)\leq \frac{(1+\zeta)M}{c}, \;
U(t) \leq M_{1}
 \; \mbox{and}\;  R(t)   \leq  M_{2}.
\end{equation*}
with
\begin{eqnarray*}
   &&  p_{0} = \bigg(r_T - d + \bigg((r_T - d)^2 +
\frac{4sr_T}{T_\textrm{max}}\bigg)^{\frac{1}{2}}\bigg)\frac{T_\textrm{max}}{2r_T}
> 0, \; M_{1}= \max\{U(0), U_\textrm{max}\} ; \\
   && \;\; M_{2}=\frac{(1+\zeta)\alpha}{\sigma}(1-\epsilon)
      U_\textrm{max} \;\;\mbox{and}\;\; M=(1+\zeta)p_0M_2.
\end{eqnarray*}
\end{theorem}
\begin{proof}
We first claim that for any solution there is a time $ t_{1}\geq 0$
such that $U(t_{1}) \leq U_\textrm{max}$. Either $U(0) \leq
U_\textrm{max} $ in which case we can take $t_{1} =0 $ or $U(0) >
U_\textrm{max}$. In the latter case let
$$ t_{2}= \sup\{t \geq 0\, :\, U(t)> U_\textrm{max}\} .$$
On interval $(0, t_{2})$ we have $ U(t) \leq U(0)e^{-\gamma t} $. It
follows that $t_{2} < +\infty $ and we can take $t_{1}=t_{2}$. At
any time where $U=U_\textrm{max}$ the derivative of $U$ is negative.
Thus $U$ becomes less than $U_\textrm{max}$ for $t$ slightly greater
than $t_{1}$ and it can never again reach the value
$U_\textrm{max}$. Hence $U(t) < U_\textrm{max} $ for all $t > t_{1}
$. In particular $U$ is globally bounded by $$ \max\{U(0),
U_\textrm{max}\}  .$$
\\\indent
We can now go with this information to the equation for $R$. From
equation (\ref{fullICCImodel5}) we obtain the differential
inequality
\begin{equation}\label{a31}
    \frac{d R }{dt} \leq \alpha(1-\epsilon)U_\textrm{max}-\sigma R.
\end{equation}
If we compare this differential inequality with the corresponding
differential equation we can see that
\begin{equation*}
   \limsup_{t \to +\infty }
   R(t) \leq  \frac{\alpha}{\sigma}(1-\epsilon)U_\textrm{max} .
\end{equation*}
In particular this proves that $R$ is globally bounded.
\\\indent
Now adding the first two equations of system (\ref{a6}), yields
\begin{equation}\label{a38}
   \frac{d(T+I)}{dt} = s + r_{T}T\left(1-\frac{T+I}{T_\textrm{max}}\right)+r_{I}I\left(1-\frac{T+I}{T_\textrm{max}}\right)-dT-\delta
   I.
\end{equation}
Since $d \leq \delta$ and $r_I \leq r_T$, it follows that
\begin{equation}\label{a388}
  \frac{d(T+I)}{dt}\leq s + r_{T}(T+I)\left(1-\frac{T+I}{T_\textrm{max}}\right) - d(T+I).
\end{equation}
Setting $P = T+I $ the differential inequality (\ref{a388}) becomes
\begin{equation}\label{a39}
 \frac{dP}{dt}\leq s+(r_{T}-d)P - \frac{r_{T}}{T_\textrm{max}}P^{2}.
\end{equation}
The right hand side of (\ref{a39}) has a unique positive root given
by
\begin{equation*}
    p_0 = \left(r_T - d + \left((r_T - d)^2 +
\frac{4sr_T}{T_\textrm{max}}\right)^{\frac{1}{2}}\right)\frac{T_\textrm{max}}{2r_T}
> 0.
\end{equation*}
Comparing a solution of (\ref{a39}) with a solution of the
corresponding differential equation gives
\begin{equation}\label{a47}
\limsup_{t \to +\infty} P(t) \leq p_0.
\end{equation}
This proves that $T$ and $I$ are globally bounded.
\\\indent
Now consider the third equation of system (\ref{a6}). We have
\begin{eqnarray*}
  \frac{dV}{dt} &=& \rho IR-cV-\frac{bTV}{T+I}, \\
   &\leq& \rho RI-cV, \\
   &\leq& M-cV,
\end{eqnarray*}
where $M = (1+\zeta)p_0M_2$. Thus for all $t> 0$ :
\begin{equation}\label{a48}
    \frac{dV(t)}{dt} + cV(t) \leq M.
\end{equation}
Solving (\ref{a48}) yields
\begin{equation}\label{a50}
    V(t) \leq V(0)\exp(-ct) + \frac{M}{c}.
\end{equation}
We deduce that,
\begin{equation*}
\limsup_{t\rightarrow +\infty} V(t) \leq \frac{M}{c}.
\end{equation*}
This proves that $V$ is globally bounded and completes the proof of
Theorem~\ref{theoglobound}.
\end{proof}
\begin{remark}
From the above results we have that any solution of the initial
value problem (\ref{a6}), (\ref{a7}) enters the region :
\begin{eqnarray*}
&&\Omega =\left\{(T, I, V, U, R)\in \mathbb{R}^{5}_{+} :
0<T(t)+I(t)\leq
2p_{0},\; \right.\\
&&\left.  0<V(t)\leq\frac{(1+\zeta)M}{c}, \; 0< U(t)\leq M_{1},\;0<
R(t)\leq M_{2} \right\},
\end{eqnarray*}
and remains there.
\end{remark}
\section{Stability Analysis of the full ICCI model }\label{sec3}
\subsection{Equilibria and the basic reproduction numbers}
Consider the equilibria of the ICCI model. One important case is
that where $V=0$, i.e. no virus is present. In a steady state with
$V=0$ it follows from (\ref{fullICCImodel3}) that $R=0$ or $I=0$.
Consider first the possibility that $I\ne 0$. Then
(\ref{fullICCImodel4}) implies that
$\left(1-\frac{T+I}{T_\textrm{max}}\right)=\frac{\delta}{r_I}$.
Together with (\ref{fullICCImodel1}) this implies that
$s=dT\left(1-\frac{\delta r_T}{dr_I}\right)$. Under the given
assumptions on the parameters this is a contradiction. Hence in fact
$I=0$. It then follows from equation (\ref{fullICCImodel1}) that $T$
is equal to the quantity $p_0$ introduced in the statement of
Theorem \ref{theoglobound}. The quantities $R$ and $T$ are only
constrained by the equations (\ref{fullICCImodel4}) and
(\ref{fullICCImodel5}). In one solution $U=R=0$ and the only other
possible solution has the explicit form
\begin{equation*}
U^{*}=U_1^*=U_\textrm{max}\bigg(1 -
\dfrac{1}{\mathcal{R}_{0}^{'}}\bigg)\  ,\
R^*=R_1^*=U_\textrm{max}\dfrac{\gamma}{\beta}\bigg(
\mathcal{R}_{0}^{'} - 1\bigg),
\end{equation*}
where $\mathcal{R}_{0}^{'} = \dfrac{\alpha\beta(1 -
\epsilon)}{\gamma\sigma}$. A positive solution of this type exists
precisely when ${\mathcal{R}_{0}^{'}}>1$. The two virus-free
equilibria are $E_0 = (p_0,0,0,0,0)$ and
$E_0'=(p_0,0,0,U_1^*,R_1^*)$.

At this point it is appropriate to comment on the notation
$\mathcal{R}_{0}^{'}$ just introduced. It is an example of a concept
often used in models for infection called the basic reproduction
number. Intuitively the basic reproduction number $\mathcal{R}_{0}$
is defined as the average number of secondary infections that occur
when one infective is introduced into a completely susceptible host
population \cite{Diek, Dietz, vander}. Note that $\mathcal{R}_{0}$
is also called the basic reproduction ratio \cite{Diek} or basic
reproductive rate \cite{12}. It is implicitly assumed that the
infected outsider is in the host population for the entire
infectious period and mixes with the host population in exactly the
same way that a population native would mix. A rigorous mathematical
definition of the basic reproductive number and a method for
calculating it are given in \cite{vander}. In fact this quantity is
not a feature of a system of ODE as a whole but of a boundary
equilibrium of such a system. Since we have just shown that in
general the model
 (\ref{a6}) has
two boundary equilibria it also has two basic reproduction numbers
associated to it. That associated to $E_0$ is
${\mathcal{R}_{0}^{'}}$. It is referred to in \cite{guedj2010} as
the intracellular basic reproductive number. It defines the critical
threshold of antiviral effectiveness for intracellular virus
stability. The other, that associated to $E_0'$, is what is referred
to in \cite{guedj2010} as the composite basic reproductive number.
It is given by
\begin{equation*}
\mathcal{R}_{0}''= \frac{b\rho \overline{R}}{(b + c)\big(\delta -
r_{I} \big(1-\frac{p_0}{T_\textrm{max}}\big)\big)}.
\end{equation*}

It defines the critical threshold of antiviral effectiveness for
extra-cellular virus stability, where
$\overline{R}=U_\textrm{max}\left(\frac{\alpha}{\sigma}-\frac{\gamma}{\beta}
\right)$ is the pre-treatment steady-state value for $R$.

\subsection{The existence of infected equilibria}
In this subsection, we investigate the existence of infected
equilibria of the ODE system (\ref{a6}). Thus, let $E^* = (T^*, I^*,
V^*, U^*, R^*)$ be an equilibrium point with infection, where
 $T^*>0$ ,\;  $I^*>0$, \; $V^*>0$ ,\; $U^*>0$,
 \; $R^*>0$. Note that a non-negative steady state automatically satisfies
 $T^*>0$ and can only satisfy $I^*=0$ or $R^*=0$ if $V^*=0$. In the latter case it is one
of the virus-free steady states considered above. Thus any
non-negative steady state other than the virus-free steady states is
positive. It satisfies the following two algebraic systems:
\begin{equation}\label{p3} \left\{
    \begin{array}{ll}
       s + r_{T}T^* \Bigg(1 - \dfrac{T^* + I^*}{T_\textrm{max}}\Bigg) - dT^* - \dfrac{bT^*V^*}{T^*+I^*} = 0, & \\
       r_{I}I^*\Bigg(1 - \dfrac{T^* + I^*}{T_\textrm{max}}\Bigg) + \dfrac{bT^*V^*}{T^*+I^*  } - \delta I^* = 0, &  \\
       \rho R^*I^*  - cV^* - \dfrac{bT^*V^*}{T^* + I^*} = 0,  & \\
    \end{array}
  \right.
  \end{equation}
and
\begin{equation}\label{p4}
\left\{
    \begin{array}{ll}
       \beta R^* \Bigg(1 - \dfrac{U^*}{U_\textrm{max}}\Bigg) - \gamma U^* = 0, \\
       \alpha (1 - \epsilon)U^* - \sigma R^* = 0.
    \end{array}
  \right.
  \end{equation}
Note that (\ref{p4}) is decoupled from (\ref{p3}). Its unique
positive solution, which only exists when $\mathcal{R}_{0}^{'} > 1$,
was given in the last section. It remains to solve (\ref{p3}) after
substituting in the value of $R^*$ given by that positive solution.

Now let $\mu = \rho R^* $ and $X = \dfrac{T^*}{T^* + I^*}$. Thus,
since $T^*$ and $I^*$ are positive, it follows that $0 < X < 1$. The
system (\ref{p3}) can be rewritten in the form
\begin{equation}\label{p5} \left\{
    \begin{array}{ll}
       s + r_{T}T^* \Bigg(1 - \dfrac{T^*}{T_\textrm{max}X}\Bigg) - dT^* - bV^*X = 0, & \\
       r_{I}I^*\Bigg(1 - \dfrac{T^*}{T_\textrm{max}X}\Bigg) + bV^*X - \delta I^* = 0, &  \\
       \mu I^*  - cV^* - bV^*X = 0.  & \\
    \end{array}
  \right.
  \end{equation}
If a positive steady state $(T^*,I^*,V^*)$ is given a corresponding
value of $X$ can be calculated. Conversely, under an additional
condition introduced below, $T^*$, $I^*$ and $V^*$ can be expressed
in terms of $X$, as will now be shown. Solving the last equation of
(\ref{p5}) with respect to  $V^*$ gives
\begin{equation}\label{p6}
V^* = \dfrac{\mu}{c + bX}I^*.
\end{equation}
Substituting (\ref{p6}) into the second equation of (\ref{p5})
yields
\begin{equation*}
 r_{I}I^*\left(1 - \dfrac{T^*}{T_\textrm{max}X}\right) + b\dfrac{\mu}{c + bX}I^*X - \delta I^* =
 0;
\end{equation*}
since $I^* \neq 0 $, the previous equation implies that
\begin{equation*}
 r_{I}\left(1 - \dfrac{T^*}{T_\textrm{max}X}\right) + \dfrac{b\mu X}{c + bX} - \delta =
 0.
\end{equation*}
Thus,
\begin{equation}\label{p81}
T^* = \frac{T_\textrm{max}X\bigg[b\mu X + (c + bX)(r_I - \delta)
\bigg]}{r_I(c + bX)}.
\end{equation}
A sufficient condition for the positivity of the right hand side of
(\ref{p81}) is that $r_I-\delta\ge 0$ and this assumption will be
made from now on. We are not aware whether the right hand side of
(\ref{p81}) is always positive in the absence of this assumption.
Having calculated $T^*$ in terms of $X$ we can calculate $I^*$ using
the relation $I^*=\frac{T^*(1-X)}{X}$ and $V^*$ using (\ref{p6}).
Under the assumption $0<X<1$ these quantities $(T^*,I^*,V^*)$ are
positive. When do quantities defined in this way in terms of
$X\in(0,1)$ define a steady state of (\ref{p3})? The equation
defining $I^*$ shows that the equation originally used to define $X$
holds. It follows from (\ref{p6}) that the third equation of
(\ref{p5}) holds and this implies the third equation of (\ref{p3}).
The defining equation for $T^*$ together with (\ref{p6}) implies
that the second equation of (\ref{p3}) holds. Substituting the
expression for $(T^*,I^*,V^*)$ into the first equation of (\ref{p5})
and multiplying by $r_I^2(c+bX)^2$ gives
\begin{eqnarray}
&&sr_I^2(c+bX)^2
+ T_\textrm{max}X\bigg[b(\mu+r_I - \delta)X+ c(r_I - \delta) \bigg]\nonumber\\
&&\times\Bigg[b(\delta r_{T}-dr_I-\mu (r_{T}-r_I))X+c(\delta
r_{T}-dr_I) -b\mu r_I\Bigg]=0
\end{eqnarray}
We see that under the assumption $r_I-\delta\ge 0$ positive steady
states are in one to one correspondence with roots of a cubic
polynomial $p(X)$ in the interval $(0,1)$, where
$$p(X)=a_3X^3+a_2X^2+a_1X+a_0$$
and
\begin{eqnarray}
&&a_0=c^2sr_I^2,\nonumber\\
&&a_1=2bcsr_I^2+cT_\textrm{max}(r_I-\delta)
(c(\delta r_{T}-dr_I)-b\mu r_I),\nonumber\\
&&a_2=b^2sr_I^2+bT_\textrm{max}[(\mu+r_I - \delta)(c(\delta
r_{T}-dr_I)-b\mu r_I)
\nonumber\\
&& \;\;\;\;\; +c(r_I - \delta)(\delta r_{T}-dr_I-\mu (r_{T}-r_I))],\nonumber\\
&&a_3=b^2T_\textrm{max}(\mu+r_I - \delta) (\delta r_{T}-dr_I-\mu
(r_{T}-r_I)). \nonumber
\end{eqnarray}
The roots of the polynomial $p$ depend continuously on the
parameters. If the parameters vary in a compact set then the roots
cannot approach $X=0$ since $a_0$ is bounded away from zero. On the
other hand the roots might approach $X=1$. Consider a sequence in
parameter space which converges to a positive limit and a sequence
of roots $X_n\in (0,1)$ of $p$ corresponding to these parameter
values with $\lim\limits_{n\to\infty} X_n=1$. Let
$(T^*_n,I^*_n,V^*_n)$ be the corresponding sequence of positive
steady states. $T^*_n$ converges to a positive limit. It follows
that $I^*_n\to 0$ and $V^*_n\to 0$. Thus this sequence of steady
states converges to a steady state on the boundary. We know the
steady states on the boundary explicitly. Since in this limit two
steady states approach each other the steady state in the limiting
case must be degenerate. It will be shown in the next subsection
that this can only happen when one of the basic reproduction numbers
is one. Consider now a convergent sequence of parameters for which
both reproduction numbers remain strictly greater than one. Then the
corresponding sequence of roots of $p$ remains in a compact subset
of $(0,1)$. It follows that the number of roots of the polynomial,
counting multiplicity, is independent of the parameters in this
region modulo two.

Consider next what happens if $s$ tends to zero while the other
parameters are held fixed. The polynomial $p$ converges to the
product of $X$ with a quadratic polynomial $q$ and the the values of
the roots can be read off. When $s=0$ we have
$p_0=\frac{(r_T-d)T_\textrm{max}}{r_T}$. Under the assumption that
$r_I>\delta$ one of the roots of $q$ is negative. The second factor
in the denominator of the expression for $\mathcal{R}_0''$ can be
bounded below by $\delta\left(1-\frac{r_Id}{r_T\delta}\right)$,
which is positive. Thus $\mathcal{R}_0''$ is positive. If we make
$\alpha$ large while fixing all other parameters then
$\mathcal{R}_0'$ and $\mathcal{R}_0''$ can be made as large as
desired, in particular greater than one. In this situation $\mu$
also becomes arbitrarily large and this implies that the other root
of $q$ is also negative. $q$ tends to $-\infty$ when $|X|$ is large
and $q(0)<0$. Hence $p'(0)<0$. It can be concluded that under these
circumstances for $s$ small and positive, where $p'(0)$ remains
negative but $p(0)>0$, the polynomial $p$ has precisely one root in
$(0,1)$ and there exists precisely one infected steady state. Since
we have now shown that there are points in this region of parameter
space where this number is one it follows that it is always odd. In
particular there always exists at least one positive steady state
under these assumptions. There are always one, two or three positive
steady states but we will not answer the question of whether there
can be more than one in this paper.
\\
Let us study the local stability of the uninfected equilibrium
$E_0$.
\subsection{Local stability of HCV uninfected equilibria}
\begin{proposition}\label{propplocstau}
If $\mathcal{R}_0^{'} < 1$, then the uninfected equilibrium  $E_0$
of the ODE model  (\ref{a6}) is locally asymptotically stable.
\end{proposition}
\begin{proof}
The Jacobian matrix at $E_0$ of the ODE model  (\ref{a6}) is given
by
\[  J(E_{0})=\left(
                       \begin{array}{ccccc}
                         -a_{11}& -a_{12} & -b & 0 & 0 \\
                         0     & a_{22} & b & 0 & 0 \\
                         0     & 0 & -c-b & 0 & 0 \\
                         0     & 0 & 0 & -\gamma & \beta \\
                         0     & 0 & 0 & \alpha(1-\epsilon) & -\sigma \\
                       \end{array}
                     \right),
\]
where \[a_{11}=\sqrt{(r_{T}-d)^{2}+\frac{4sr_{T}}{T_\textrm{max}}}>0
\,\,,\,\, a_{12}=\frac{p_{0}}{T_\textrm{max}}r_{T}>0\,\, , \,
a_{22}=-\Big[\delta+r_{I}\left(\frac{p_{0}}{T_\textrm{max}}-1\right)\Big]<0.
\]
The characteristic polynomial $P_J$ associated to $J(E_{0})$ is
given by
$$P_{J}(X)=(-a_{11}-X)( a_{22}-X)(-c-b-X)\Big( X^{2} +
(\gamma+\sigma)X +\gamma\sigma-\alpha (1-\epsilon)\beta \Big).$$

Since $-a_{11}<0$ , $a_{22}<0$ and $-(b+c)<0$, the real part of the
roots of  $P_{J}$ are negative if and only if the roots of the
quadratic polynomial defined by :
\begin{equation}
     T(X)= X^{2} + (\gamma+\sigma)X +\gamma\sigma-\alpha(1-\epsilon)
\beta
\end{equation}
have negative real part. Applying the Routh-Hurwitz criterion to the
previous quadratic equation, the roots of $T(X)$ have negative real
part if and only if
\begin{equation}
    0<\gamma\sigma-\alpha(1-\epsilon) \beta,
\end{equation}
which is equivalent to
\begin{equation*}
    1 > \frac{\alpha\beta(1-\epsilon)}{\gamma\sigma},
\end{equation*}
 i.e ,
$$\mathcal{R}_{0}^{'}<1.$$
which completes the proof of the proposition~\ref{propplocstau}.
\end{proof}

\begin{proposition}\label{proplocstab2}
If $\mathcal{R}_0^{'} > 1$ and $\mathcal{R}_0^{''} < 1$, then the
second uninfected equilibrium $E'_0$ of the ODE model  (\ref{a6}) is
locally asymptotically stable.
\end{proposition}
\begin{proof}
The Jacobian matrix associated to the ODE model (\ref{a6}) at
$E'_{0}$ is given by
\[  J(E'_{0})=\left(
                       \begin{array}{ccccc}
                         -a_{11}& -a_{12} & -b & 0 & 0 \\
                         0     & a_{22} & b & 0 & 0 \\
                         0     & \rho R^{*} & -c-b & 0 & 0 \\
                         0     & 0 & 0 & -\beta e_{44}-\gamma & \beta(1-e_{45}) \\
                         0     & 0 & 0 & \alpha(1-\epsilon) & -\sigma \\
                       \end{array}
                     \right),
\]
where
\[e_{44}=\frac{R^{*}}{U_\textrm{max}}>0
\,\,,\,\, e_{45}=\frac{U^{*}}{U_\textrm{max}}>0,
\] and $a_{11}$, $a_{12}$, $a_{22}$ are defined in the same way as
in the previous proof. The characteristic polynomial $P_J(X)$
associated to $J(E'_{0})$ is the product of a quadratic polynomial
generalizing the polynomial $T$ introduced above with a cubic
polynomial. The signs of the coefficients in the quadratic
polynomial remain the same and thus its roots have negative real
parts. The cubic contains a factor $X+a_{11}$ and so to show that
all eigenvalues of the linearization have negative real part it
suffices to control the roots of the remaining quadratic polynomial.
It is given by
\begin{equation*}
X^2+(-a_{22}+b+c)X-a_{22}(b+c)-b\rho R^*.
\end{equation*}
The coefficient of $X$ is negative and so it is enough to show that
the constant term is positive.
\begin{equation*}
  -a_{22}(b+c)-b\rho R^*=\rho\left(\frac{1}{\mathcal{R}_0''-1}\right)b\rho\bar R
  +bp(\bar R-R^*)>0
\end{equation*}
when $\mathcal{R}_0''<1$ since $\bar R>R^*$. This completes the
proof of the proposition~\ref{proplocstab2}.
\end{proof}
Now let us study the local stability of infected equilibria.
\subsection{Local stability of HCV infected equilibria}
\begin{proposition}\label{propplocstai}
  If $\mathcal{R}_0^{'} > 1$, $\mathcal{R}_0^{''} > 1$, $s$ is sufficiently
  small and $\alpha$ and $b$
  sufficiently large then the  unique equilibrium  with infection
$E^{*}$ of the ODE model (\ref{a6}) is locally asymptotically
stable.
\end{proposition}
\begin{proof}
The Jacobian matrix at $E^*$ of the ODE model  (\ref{a6}) is given
by :
\[  J(E^*)=\left(
                       \begin{array}{ccccc}
                            -a_{1}& b_{1} & -\frac{bT^{*}}{T^{*}+I^{*}} & 0     & 0 \\
                            a_{2}& -b_{2} & \frac{bT^{*}}{T^{*}+I^{*}} & 0     & 0 \\
                            -\frac{bI^{*}V^{*}}{(T^{*}+I^{*})^{2}}& \rho R^{*}+\frac{bT^{*}V^{*}}{(T^{*}+I^{*})^{2}} & -c-\frac{bT^{*}}{T^{*}+I^{*}} & 0     & \rho I^{*} \\
                              0  & 0     & 0     & -\frac{\beta R^{*}}{U_\textrm{max}}-\gamma & \beta\Bigg(1-\frac{U}{U_\textrm{max}} \Bigg) \\
                              0  & 0     & 0     & \alpha(1-\epsilon) & -\sigma \\
                       \end{array}
                     \right),
\] where $$ \left\{
            \begin{array}{ll}
              a_{1}=d-r_{T}+\frac{2r_{T}T^{*}}{T_\textrm{max}}+\frac{r_{T}I^{*}}{T_\textrm{max}}
              +\frac{bI^{*}V^{*}}{(T^{*}+I^{*})^{2}}, & \hbox{} \\
              a_{2}=-\frac{r_{I}I^{*}}{T_\textrm{max}}+\frac{bI^{*}V^{*}}{(T^{*}+I^{*})^{2}}, & \hbox{} \\
              b_{1}=-\frac{r_{T}T^{*}}{T_\textrm{max}}+\frac{bT^{*}V^{*}}{(T^{*}+I^{*})^{2}}, & \hbox{} \\
              b_{2}=\delta-r_{I}+\frac{r_{I}T^{*}}{T_\textrm{max}}+ \frac{2r_{I}I^{*}}{T_\textrm{max}}+\frac{bT^{*}V^{*}}{(T^{*}+I^{*})^{2}}. &\hbox{}
            \end{array}
          \right.
$$
This matrix is block triangular and thus its characteristic
polynomial is the product of those of the top left $3\times 3$
matrix and the bottom right $2\times 2$ matrix. The latter is equal
to a polynomial we studied in the previous case and thus its roots
have negative real parts. It remains to analyse the other factor,
call it
\begin{equation}
P_1(X)=X^{3}+\lambda_{2}X^{2}+\lambda_{1}X+\lambda_{0}\label{polyinfone}
\end{equation}
where
 \begin{eqnarray*}
\lambda_{2}
&=&(\delta-r_{I})+(d-r_{I})+c+\frac{bT^{*}}{T^{*}+I^{*}}+\frac{bV^{*}I^{*}}{(T^{*}+I^{*})^{2}}
+\frac{(T^{*}+I^{*})(r_{T}+r_{I})+r_{T}T^{*}+r_{I}I^{*}}{T_\textrm{max}}
\\
&& +\frac{bT^{*}V^{*}}{(T^{*}+I^{*})^{2}},
\end{eqnarray*}
\begin{eqnarray*}
\lambda_{1}
&=&c\left(d-r_{T}+\frac{2r_{T}T^{*}}{T_\textrm{max}}+\frac{r_{T}I^{*}}{T_\textrm{max}}
+\frac{bI^{*}V^{*}}{(I^{*}+T^{*})^{2}}
 \right)+\left(\delta-r_{I}+\frac{r_{I}T^{*}}{T_\textrm{max}}+\frac{2 r_{I}I^{*}}{T_\textrm{max}}\right)
 \bigg(c+\frac{bT^{*}}{T^{*}+I^{*}}+d-r_{T}\\
    &&+\frac{2 r_{T}T^{*}}{T_\textrm{max}}
    +\frac{bV^{*}I^{*}}{(T^{*}+I^{*})^{2}}\bigg) +\frac{bT^{*}}{T^{*}+I^{*}}
   \left( d-r_{T}+\frac{2r_{T}T^{*}}{T_\textrm{max}}+\frac{r_{T}I^{*}}{T_\textrm{max}}\right)
   +\frac{r_{T}I^{*}}{T_\textrm{max}}\left(\delta-r_{I}+\frac{2r_{I}T^{*}}{T_\textrm{max}}\right)\\
   &&\frac{b r_{I}T^{*}I^{*}V^{*}}{T_\textrm{max}(T^{*}+I^{*})^{2}} +\frac{bT^{*}V^{*}}{(T^{*}+I^{*})^{2}}
  \left(\rho R^{*}+d-r_{T}+\frac{2r_{T}T^{*}}{T_\textrm{max}}+\frac{ r_{T}I^{*}}{T_\textrm{max}}\right)\\
   &&+\frac{b
   r_{T}T^{*}I^{*}V^{*}}{T_\textrm{max}(T^{*}+I^{*})^{2}}-\frac{\rho
   R^{*}bT^{*}}{T^{*}+I^{*}},
\end{eqnarray*}
\begin{eqnarray*}
\lambda_{0}&=&c\left(d-r_{T}+\frac{r_{T}(2T^{*}+I^{*})}{T_\textrm{max}}+\frac{bI^{*}V^{*}}{(T^{*}+I^{*})^{2}}
\right) \left( \delta-r_{I}+\frac{r_{I}(T^{*} + 2I^{*})}{T_\textrm{max}}\right)\\
 && +\frac{bT^{*}}{T^{*}+I^{*}}\left(d-r_{T}+\frac{r_{T}(2T^{*} + I^{*})}{T_\textrm{max}}
\right)\left( \delta-r_{I}+\frac{r_{I}(T^{*} + 2I^{*})}{T_\textrm{max}}\right)\\
&&+\frac{r_{I}I^{*}}{T_\textrm{max}}\left(c+\frac{bT^{*}}{T^{*}+I^{*}}\right)\left(\frac{bT^{*}V^{*}}
{(T^{*}+I^{*})^{2}}-\frac{r_{T}T^{*}}{T_\textrm{max}} \right)
-\frac{b^{2}r_{T}T^{*}I^{*}V^{*}}{T_\textrm{max}(T^{*}+I^{*})^{2}}\\
           &&+ \frac{bcT^{*}V^{*}}{(T^{*}+I^{*})^{2}}\left(d-r_{T}+\frac{2r_{T}I^{*}}{T_\textrm{max}}
            +\frac{r_{T}T^{*}}{T_\textrm{max}}\right)
            +\frac{r_{T}bcT^{*}I^{*}V^{*}}{T_\textrm{max}(T^{*}+I^{*})^{2}}
            + \frac{b^{3}I^{*}(T^{*}V^{*})^{2}}{(T^{*}+I^{*})^{5}}(T^{*}+I^{*}-1)\\
           && +\frac{\rho
           R^{*}T^{*}}{T^{*}+I^{*}}\left(\frac{I^{*}V^{*}}{(T^{*}
           +I^{*})^{2}}-1\right)- \frac{\rho r_{I}R^{*}T^{*}I^{*}}{T_\textrm{max}(T^{*}+I^{*})}.
\end{eqnarray*}
The expressions for the coefficients $\lambda_i$ are so complicated
that we have not succeeded in analyzing them in general. Instead we
concentrate on obtaining information in the limiting regime in which
the existence of a unique steady state was obtained, i.e. that where
$s$ is small. When $s$ tends to zero the steady state tends to the
point where the coordinates take the values $T^*=0$,
$I^*=\frac{(r_I-\delta)T_\textrm{max}}{cr_I}$ and $V^*=\frac{\rho
R^*(r_I-\delta)T_\textrm{max}}{cr_I}$. It will be shown that under
certain conditions the eigenvalues of the linearization about this
point with $s=0$ all have negative real parts. It then follows by
continuity that the same is true for the steady state with $s$
positive and sufficiently small. When $s=0$ the coefficients have
the following forms
\begin{eqnarray}
  &&\lambda_2=(\delta-r_I)+(d-r_T)+c+\frac{bV^*}{I^*}
     +\frac{(r_T+2r_I)I^*}{T_\textrm{max}}\\
  &&\nonumber \lambda_1=c\left(d-r_T+\frac{r_TI^*}{T_\textrm{max}}+\frac{bV^*}{I^*}\right)
     +\left(\delta-r_I+\frac{2r_II^*}{T_\textrm{max}}\right)\\
      && \hspace{1cm} \times\left(c+d-r_T+\frac{bV^*}{I^*}\right)\\
  &&\lambda_0=c\left(d-r_T+\frac{r_TI^*}{T_\textrm{max}}+\frac{bV^*}{I^*}\right)
     \left(\delta-r_I+\frac{2r_II^*}{T_\textrm{max}}\right)
\end{eqnarray}
Note that $\delta-r_I+\frac{2r_II^*}{T_\textrm{max}}=r_I-\delta>0$.
Now choose values of the parameters such that a unique steady state
exists. Then make $b$ large while fixing all the other parameters.
Then for $b$ large $\lambda_2$, $\lambda_1$ and $\lambda_0$ are all
positive for $b$ large and grow like a constant multiple of $b$. The
combination $\lambda_1\lambda_2-\lambda_0$ is positive for $b$ large
and grows like a constant multiple of $b^2$. Applying the
Routh-Hurwitz criterion this completes the proof of
proposition~\ref{propplocstai}.
\end{proof}
\section{Global stability analysis of the full ICCI model (\ref{a6})}
\label{sec4}
\begin{theorem}\label{globaatatu}
Under the conditions $M_{2}\leq \mathcal{R}_{0}^{'} \leq 1$ and
$\mathcal{R}_{0}^{''}\leq \tau_0 \leq 1$, where
    $$\tau_0 = \frac{b\rho \overline{R}}{(b + c)(\delta -
   r_{I})} $$ is a
positive constant, the uninfected equilibrium point $E_0$ of the
full ICCI ODE model (\ref{a6})is globally asymptotically stable in
the positively-invariant region $\Omega $.
\end{theorem}
\begin{proof}
Consider the Lyapunov function defined on $\mathbb{R}_{+}$ by
  \begin{equation*}
    L(t) = \rho \overline{R}I(t) + (\delta - r_{I})V(t) + \beta R(t) + \sigma
U(t).
\end{equation*}
$L$ is defined, continuously differentiable and positive definite
for all $T>0$, $I>0$, $V>0$, $U>0$, $R>0$. It is easy to see that
$L$ reaches its global minimum when the solution is at the
infection-free equilibrium $E_0$. Further, the function $L$, along
the solutions of system (\ref{a6}), satisfies :
\begin{eqnarray*}
\frac{dL}{dt} &=& \rho \overline{R}r_{I}I\Bigg(1 - \dfrac{T +
I}{T_\textrm{max}}\Bigg) + \rho \overline{R}\dfrac{bTV}{T+I} - \rho
\overline{R}\delta I + \rho RI(\delta - r_{I})
 - cV(\delta - r_{I})  \\
   &&- b(\delta - r_{I})\dfrac{TV}{T+I} +  \alpha\beta(1 - \epsilon)U - \sigma\beta R +
\sigma\beta R \bigg( 1 - \frac{U}{U_\textrm{max}}\bigg)  -
\sigma\gamma U,
\\
  &\le& \rho \overline{R}\dfrac{bTV}{T+I} - cV(\delta - r_{I}) - b(\delta - r_{I})\dfrac{TV}{T +
   I}  - \rho \overline{R}r_{I}I\dfrac{T + I}{T_\textrm{max}} + (R - \overline{R})\rho I(\delta -
   r_{I})\\
   &&+\ \alpha\beta(1 - \epsilon)U - \sigma\beta R +
\sigma\beta R \bigg( 1 - \frac{U}{U_\textrm{max}}\bigg)  -
\sigma\gamma U.
\end{eqnarray*}
Since $\frac{T}{T + I} \leq 1$, then $-V \leq -\frac{TV}{T + I}$.
Moreover, in the positively invariant set, $R(t)  \leq  M_{2} $.
Thus,
\begin{eqnarray*}
  \frac{dL}{dt} &\leq& \rho \overline{R}\dfrac{bTV}{T+I} - c(\delta - r_{I})\dfrac{TV}{T+I} - b(\delta - r_{I})\dfrac{TV}{T +
   I}  - \rho \overline{R}r_{I}I\dfrac{T + I}{T_\textrm{max}}   \\
   & & +\  ( M_{2} - \overline{R})\rho I(\delta -
   r_{I}) + \alpha\beta(1 - \epsilon)U - \sigma\beta R +
\sigma\beta R \bigg( 1 - \frac{U}{U_\textrm{max}}\bigg)  -
\sigma\gamma U.
\end{eqnarray*}
Since $\mathcal{R}_{0}^{'} \leq \overline{R}$, $-\overline{R} \leq
-\mathcal{R}_{0}^{'}$. As $1 - \frac{U}{U_\textrm{max}} \leq 1$ we
have
\begin{eqnarray*}
  \frac{dL}{dt} &\leq& \dfrac{TV}{T+I}(b\rho \overline{R} -
  (b + c)(\delta - r_{I})) + \rho I( M_{2}
   - \mathcal{R}_{0}^{'})(\delta - r_{I}),\\
   &&+\ \alpha\beta(1 - \epsilon)U - \sigma\beta R + \sigma\beta
    R - \sigma\gamma U, \\
                &\leq& \dfrac{TV}{T+I}(b\rho \overline{R} -
   (b + c)(\delta - r_{I}))  + \rho I(M_{2} - \mathcal{R}_{0}^{'})(\delta - r_{I}), \\
   &&+\ [\alpha\beta(1 - \epsilon) - \sigma\gamma]U,\\
   &\leq& (b + c)(\delta - r_{I})\dfrac{TV}{T+I}
   \bigg( \frac{b\rho \overline{R}}{(b + c)(\delta - r_{I})}
    - 1\bigg) + \bigg(\frac{\alpha\beta(1 - \epsilon)}
    {\sigma\gamma} - 1
       \bigg)\sigma\gamma U \\
       && +  \rho I( M_{2} - \mathcal{R}_{0}^{'})(\delta - r_{I}).
\end{eqnarray*}
It follows that
\begin{eqnarray*}
\frac{dL}{dt} &\leq& (b + c)(\delta - r_{I})\dfrac{TV}{T+I}\bigg( \frac{b\rho \overline{R}}{(b + c)(\delta - r_{I})} - 1\bigg) + (\mathcal{R}_0^{'} - 1) \sigma\gamma U\\
                && + \ \rho I( M_{2} - \mathcal{R}_{0}^{'})(\delta - r_{I}), \\
              &\leq& (b + c)(\delta - r_{I})\dfrac{TV}{T+I}(\tau_0 - 1)
              + (\mathcal{R}_0^{'} - 1) \sigma\gamma U  + \rho I(M_{2} - \mathcal{R}_{0}^{'})(\delta -
               r_{I}).
\end{eqnarray*}
It is clear that the condition $\tau_0 \leq 1$ and $M_{2}\leq
\mathcal{R}_{0}^{'} \leq 1$ give   $\frac{dL}{dt} \leq 0$ for all
$T>0$, $I>0$, $V>0$, $U>0$, $R > 0 $. Therefore, the largest compact
invariant subset of the set
$$  M= \left\{ (T, I, V, U, R) \in \Omega : \frac{dL}{dt} = 0\right\} $$
is the singleton $\{E_{0}\}$. By the Lasalle invariance
principle\cite{khalil}, the uninfected equilibrium point is globally
asymptotically stable if $M_{2} \leq \mathcal{R}_{0}^{'} \leq 1$ and
$\tau_0 \leq 1$. So, we obtain a sufficient condition
$\mathcal{R}_{0}^{''}\leq \tau_0$ which ensures that the
HCV-uninfected equilibrium $E_{0}$ of ODE-model system (\ref{a6}) is
globally asymptotically stable if $\tau_0 < 1$. This completes the
proof of  theorem~\ref{globaatatu}.
\end{proof}
\section{Global convergence to infected equilibria}\label{sec5}

Consider first the system
(\ref{fullICCImodel4})-(\ref{fullICCImodel5}) describing the
intracellular dynamics. Its linearization at any point of the
positive orthant has the property that both off-diagonal elements
are negative. Thus it is a competitive system and since all
solutions exist globally and are bounded it follows that any
positive solution converges to a steady state \cite{smith10}. Next
consider any positive solution of the full system
(\ref{fullICCImodel1})-(\ref{fullICCImodel5}) and any $\omega$-limit
point of that solution. There is a non-negative solution which
passes through that $\omega$-limit point. Since the projection of
the original solution onto its last two components converges to a
steady state, the limiting solution has the property that $U$ and
$R$ have the constant values $(0,0)$ or $(U^*,R^*)$. Its projection
onto the first three coordinates is a solution of the system
obtained from (\ref{fullICCImodel1})-(\ref{fullICCImodel3}) by
fixing $R$ to be equal to $0$ or $R^*$. In the case $R=0$ the
function $V$ is a Lyapunov function for the projected solution and
it must converge to zero. If we again pass to an $\omega$-limit
point and a solution passing through it we get a solution of the
system obtained by setting $V=0$ and projecting on the first two
coordinates. The resulting two-dimensional system is again
competitive and so the solution must converge to a steady state.
This system has no positive steady states and so the convergence
must be to a point of the boundary. On the boundary $I=0$ the only
steady state is given by $T=p_0$. The boundary $T=0$ consists
entirely of steady states. It is a centre manifold of any of its
points and the non-zero eigenvalue of the linearization at any of
these points is positive. Thus no solution can approach a point of
this type. It can be concluded that the original solution converges
to the point $E_0$.

It remains to analyse the case $R=R^*$. For this it suffices to
study the late time behaviour of the system
(\ref{fullICCImodel1})-(\ref{fullICCImodel3}) with $R=R^*$ and so we
will now concentrate on that system and apply the geometric approach
of \cite{liMuldowney1996}. In doing this we use the quantity
$$\tilde{T}=\frac{T_\textrm{max}}{2r_{I}}\Big[r_{I}-\delta+\sqrt{(r_{I}-\delta)^{2}
  +\frac{4sr_{I}}{T_\textrm{max}}}\Big].$$
It follows from Theorem \ref{theoglobound} that for any constant
$\zeta>0$ any solution satisfies $T+I\le (1+\zeta)p_0$ for $t$
sufficiently large. By a very similar argument to that used in the
proof of that theorem $T+I\ge(1-\zeta)\tilde T$ for $t$ large. To
prove the main result of this section we need the following lemma.
\begin{lemma}\label{lemperst}
If $\mathcal{R}''_{0} > 1 $ then the system
(\ref{fullICCImodel1})-(\ref{fullICCImodel3}) with $R=R^*$ is
uniformly persistent.
\end{lemma}
\begin{proof}
This result follows from an application of Theorem 4.3 in
\cite{freedman1994} with $X=\mathbb{R}^{3}$ and $E=\Omega$. The
maximal invariant set $M$ on the boundary $\partial\Omega$ is the
singleton $\{E_{0}\}$, and it is isolated. From Theorem~4.3 in
\cite{freedman1994} we can see that the uniform persistence of the
system (\ref{cellmod}) is equivalent to the instability of the
disease-free equilibrium $E_{0}$ . On the other hand, we have proved
in Theorem~\ref{globaatatu} that $E_{0}$ is unstable if
$\mathcal{R}''_{0} > 1 $. Thus, the system
(\ref{fullICCImodel1})-(\ref{fullICCImodel3}) with $R=R^*$ is
uniformly persistent when $\mathcal{R}''_{0} > 1 $.
\end{proof}
\begin{theorem}
Suppose that the following inequality holds
\begin{equation*}
    \xi = \max\bigg\{r_T - d + \frac{(r_I + r_T)p_0}{T_\textrm{max}}
    - \frac{r_T\tilde{T}}{T_\textrm{max}} + 2b ; r_I - \delta +
    \frac{(r_I + r_T)p_0}{T_\textrm{max}} - \frac{r_I\tilde{T}}{T_\textrm{max}}
    + 2b
   \bigg\} < 0.
\end{equation*}
Then every positive solution of the system
(\ref{fullICCImodel1})-(\ref{fullICCImodel3}) with $R=R^*$ converges
to a steady state. If $\mathcal{R}''_0>1$ then every positive
solution converges to a positive steady state.
\end{theorem}
\begin{proof}
Any positive solution of
(\ref{fullICCImodel1})-(\ref{fullICCImodel3}) is bounded and so has
a non-empty $\omega$-limit set and that set is connected. Consider
an $\omega$-limit point on the boundary of the positive orthant and
a solution passing through that point at some time. It lies entirely
in the $\omega$-limit set of the original solution and so, in
particular, is non-negative. If $V$ were zero at that point and
$I\ne 0$ then the evolution equation for $V$ would imply that the
limiting solution was negative slightly before the initial time, a
contradiction. Hence for a point of this type $V=0$ implies $I=0$.
By a similar argument the evolution equation for $T$ implies that
$T>0$ at such a point. Finally the evolution equation for $I$
implies that if $I=0$ then $V=0$. Thus the only possible
$\omega$-limit point is $(p_0,0,0)$. The set of steady states is
finite and therefore discrete. Hence if a solution does not converge
to a steady state then by connectedness its $\omega$-limit set must
contain points which are not steady states and which are contained
in the positive orthant. These are non-equilibrium non-wandering
points (see \cite{liMuldowney1996} for the terminology). This
implies that there exists a periodic solution of a system which is a
small perturbation of the original one (\cite{liMuldowney1996},
Lemma 2.1). This is impossible if the quantity $\bar q_2$ in
\cite{liMuldowney1996} is negative (\cite{liMuldowney1996}, Theorem
3.1). Thus to prove the theorem it suffices to show that the
inequality assumed as a hypothesis implies that $\bar q_2<0$. This
will now be done, following closely an argument given in
\cite{liMuldowney1996}.

The Jacobian matrix $J$ associated with a general solution to
(\ref{fullICCImodel1})-(\ref{fullICCImodel3}) is
\begin{equation*}
    J = \left(
  \begin{array}{ccc}
    p & \frac{bTV}{(T + I)^2} - \frac{r_{T}T}{T_\textrm{max}} & -\frac{bT}{T+I} \\
    \frac{bIV}{(T+I)^2} - \frac{r_{I}I}{T_\textrm{max}} & q & \frac{bT}{T+I} \\
    -\frac{bIV}{(T + I)^2} & \rho R^* + \frac{bTV}{(T + I)^2} & -c-\frac{bT}{T+I} \\
  \end{array}
\right),
\end{equation*}
with $p = r_T-d -
\frac{2r_{T}T}{T_\textrm{max}}-\frac{r_{T}I}{T_\textrm{max}}-\frac{bIV}{(T
+ I)^2}$ and $q = r_I - \delta
-\frac{r_{I}T}{T_\textrm{max}}-\frac{2r_{I}I}{T_\textrm{max}}-\frac{bTV}{(T
+ I)^2}$. The second additive compound matrix $J^{[2]}$ is
\begin{equation*}
    J^{[2]} = \left(
  \begin{array}{ccc}
    p + q & \frac{bT}{T+I} & \frac{bT}{T+I} \\
    \rho R^* + \frac{bTV}{(T + I)^2} & p -c-\frac{bT}{T+I}  & \frac{bTV}{(T + I)^2} - \frac{r_{T}T}{T_\textrm{max}}  \\
    \frac{bIV}{(T + I)^2} &\frac{bIV}{(T + I)^2} - \frac{r_{I}I}{T_\textrm{max}}  & q-c-\frac{bT}{T+I} \\
  \end{array}
\right).
\end{equation*}
We consider the  matrix $P = \textrm{diag}\big\{1, \tfrac{I}{V},
 \tfrac{I}{V}\big\}$. It follows then that
$$P_fP^{-1} = \textrm{diag}\big\{0, \tfrac{\dot{I}}{I} - \tfrac{\dot{V}}{V} ,
  \tfrac{\dot{I}}{I} - \tfrac{\dot{V}}{V} \big\}$$
with $\dot{ }=\frac{d }{dt}$ and where the matrix $P_f =
\textrm{diag}\big\{0, \tfrac{d}{dt}(\tfrac{I}{V}),
 \tfrac{d}{dt}(\tfrac{I}{V})\big\}$
 is obtained by replacing each entry $p_{ij}$ of $P$ by its
 derivative in the direction of the solution of (\ref{cellmod}).
Furthermore, we have
  \begin{equation*}
\mathcal{B} = P_fP^{-1} + PJ^{[2]}P^{-1} = \left(
                \begin{array}{cc}
                  B_{11} & B_{12} \\
                  B_{21} & B_{22} \\
                \end{array}
              \right),
  \end{equation*}
where
\begin{eqnarray*}
 B_{11} &=& p+q , \\
 B_{12} &=&  \left(
                                          \begin{array}{cc}
                                            \frac{bTV}{I(T+I)} & \frac{bTV}{I(T+I)} \\
                                          \end{array}
                                        \right)  ,\\
B_{21} &=& \left(
              \begin{array}{c}
                 \frac{\rho R^*I}{V} + \frac{bTI}{(T+I)^2} \\
                \frac{bI^2}{(T+I)^2} \\
              \end{array}
            \right), \\
B_{22} &=&   \left(
                                           \begin{array}{cc}
                                             \frac{\dot{I}}{I} - \frac{\dot{V}}{V} - c +p-\frac{bT}{T+I} & \frac{bTV}{(T + I)^2} - \frac{r_{T}T}{T_\textrm{max}} \\
                                             \frac{bIV}{(T + I)^2} - \frac{r_{I}I}{T_\textrm{max}} &  \frac{\dot{I}}{I} - \frac{\dot{V}}{V} - c +q-\frac{bT}{T+I}  \\
                                           \end{array}
                                         \right).
\end{eqnarray*}
Define the norm in $\mathbb{R}^3 $ as $\|(x, y, z)\|=\max\{|x|, |y|
+ |z|\}$ for $ (x, y, z)\in \mathbb{R}^3$. Then the Lozinskii
measure $\mu$ with respect to the norm $\| \cdot \|_{1}$ can be
estimated as follows(see \cite{MatinJr1974}) : we have
\begin{equation}\label{lozi11}
\mu(\mathcal{B}) \leq \sup\{g_1, g_2\},
\end{equation}
where :
\begin{equation*}
    g_1 = \mu_1(B_{11}) + \|B_{12}\|_1  \;\;\mbox{and} \;\; g_2 =
    \|B_{21}\|_1 + \mu_1(B_{22}).
\end{equation*}
Here $\mu_1$ denotes the  Lozinskii measure with respect to the $\|
\cdot \|_{1}$ vector norm, and $\|B_{12}\|_1$ and  $\|B_{21}\|_1$
are matrix norms with respect to the $\| \cdot \|_{1}$ norm .
Moreover, we have
\begin{equation*}
\mu_1(B_{11}) = p + q \, , \, \,\|B_{12}\|_1 = \frac{bTV}{I(T +
I)}\, , \, \, \|B_{21}\|_1 = \frac{\rho R^*I}{V} + \frac{bT}{T+I}.
\end{equation*}
 To calculate $\mu_1(B_{22})$, add
the absolute value of the off-diagonal elements to the diagonal one
in each column of $B_{22}$, and then take the maximum of two sums.
Thus, for $t$ sufficiently large,
 \begin{eqnarray*}
   \mu_1(B_{22}) &=& \max\bigg\{\frac{\dot{I}}{I} - \frac{\dot{V}}{V} - c +p-\frac{bT}{T+I} + \bigg| \frac{bIV}{(T + I)^2} - \frac{r_{I}I}{T_\textrm{max}} \bigg| ; \frac{\dot{I}}{I} - \frac{\dot{V}}{V} - c +q-\frac{bT}{T+I}  \\
    &&+ \bigg|\frac{bTV}{(T + I)^2} - \frac{r_{T}T}{T_\textrm{max}}\bigg|  \bigg\}, \\
    &\leq & \frac{\dot{I}}{I} - \frac{\dot{V}}{V} - c + \max\bigg\{ r_T-d -
\frac{2r_{T}T}{T_\textrm{max}}-\frac{r_{T}I}{T_\textrm{max}}-\frac{bIV}{(T
+ I)^2} +
\frac{bIV}{(T + I)^2} + \frac{r_{I}I}{T_\textrm{max}} ;    \\
    && r_I - \delta
-\frac{r_{I}T}{T_\textrm{max}}-\frac{2r_{I}I}{T_\textrm{max}}-\frac{bTV}{(T
+
I)^2} + \frac{bTV}{(T + I)^2} + \frac{r_{T}T}{T_\textrm{max}} \bigg\}, \\
    &\leq& \frac{\dot{I}}{I} - \frac{\dot{V}}{V} - c + \max\bigg\{ r_T-d + \frac{(r_I + r_{T})p_0}{T_\textrm{max}} - \frac{r_{T}\tilde{T}}{T_\textrm{max}} ; r_I - \delta + \frac{(r_I + r_{T})p_0}{T_\textrm{max}} -
\frac{r_{I}\tilde{T}}{T_\textrm{max}} \bigg\},
 \end{eqnarray*}
In the last inequality it has been used that any solution satisfies
$T+I\le (1+\zeta)p_0$ and $T+I\ge (1-\zeta)\tilde T$ for $t$
sufficiently large. Note that the inequality in the statement of the
theorem implies that there exists $\zeta>0$ for which the analogous
inequality holds where $p_0$ is replaced by $\bar p_0=(1+\zeta)p_0$
and $\tilde T$ by $\bar T=(1-\zeta)\tilde T$. From the second and
third equations of \eqref{cellmod}, we have
\begin{equation*}
    \frac{\dot{I}}{I} = r_{I}\Bigg(1 - \dfrac{T + I}{T_\textrm{max}}\Bigg) + \dfrac{bTV}{I(T + I)} -
    \delta,
\end{equation*}
\begin{equation*}
    \frac{\dot{V}}{V} = \frac{\rho R^*I}{V}  - c - \dfrac{bT}{T + I}.
\end{equation*}
Hence,
\begin{eqnarray*}
  g_1 &=& p + q +  \frac{bTV}{I(T +
I)},  \\
   &=& \frac{\dot{I}}{I} + r_T-d - \frac{r_{T}T}{T_\textrm{max}}
   - \frac{r_{I}I}{T_\textrm{max}} - \frac{r_{T}(T + I)}{T_\textrm{max}}
    - \frac{bV}{T + I},\\
   &\leq& \frac{\dot{I}}{I} + r_T-d + \frac{(r_I + r_{T})p_0}{T_\textrm{max}} -
   \frac{r_{T}\tilde{T}}{T_\textrm{max}},
\end{eqnarray*}
\begin{eqnarray*}
  g_2 &\leq & \frac{\rho R^*I}{V} + \frac{bT}{T+I} + \frac{\dot{I}}{I} - \frac{\dot{V}}{V} - c + \max\bigg\{ r_T-d + \frac{(r_I + r_{T})p_0}{T_\textrm{max}} - \frac{r_{T}\tilde{T}}{T_\textrm{max}} ; \\
      &&  r_I - \delta + \frac{(r_I + r_{T})p_0}{T_\textrm{max}} - \frac{r_{I}\tilde{T}}{T_\textrm{max}}
      \bigg\} \\
   &\leq & \frac{\dot{I}}{I}+ \frac{2bT}{(T+I)} + \max\bigg\{ r_T-d + \frac{(r_I + r_{T})p_0}{T_\textrm{max}} - \frac{r_{T}\tilde{T}}{T_\textrm{max}} ; r_I - \delta + \frac{(r_I +
r_{T})p_0}{T_\textrm{max}} - \frac{r_{I}\tilde{T}}{T_\textrm{max}}
      \bigg\}, \\
   &\leq& \frac{\dot{I}}{I} + 2b + \max\bigg\{ r_T-d +
    \frac{(r_I + r_{T})p_0}{T_\textrm{max}} - \frac{r_{T}\tilde{T}}
    {T_\textrm{max}} ; r_I - \delta + \frac{(r_I + r_{T})p_0}{T_\textrm{max}} -
   \frac{r_{I}\tilde{T}}{T_\textrm{max}} \bigg\}, \\
   &\leq& \frac{\dot{I}}{I} + \max\bigg\{ r_T-d +
   \frac{(r_I + r_{T})p_0}{T_\textrm{max}} - \frac{r_{T}
   \tilde{T}}{T_\textrm{max}} + 2b ; r_I - \delta + \frac{(r_I + r_{T})p_0}{T_\textrm{max}} -
   \frac{r_{I}\tilde{T}}{T_\textrm{max}} + 2b \bigg\}.
\end{eqnarray*}
Therefore,
\begin{equation*}
    \mu(\mathcal{B}) \leq \frac{\dot{I}}{I} + \xi
\end{equation*}
is valid for $ t\geq t_{1}$, where $t_{1}$ is a sufficiently large
positive constant. Along each solution $(T(t), I(t), V(t))$ of model
(\ref{cellmod}) with $(T_{0},I_{0},V_{0}) \in K $, where $K$ is the
compact absorbing set and exists by Theorem~\ref{theoglobound} and
Lemma~\ref{lemperst}, we have
\begin{eqnarray*}
  \frac{1}{t}\int_{0}^{t}\mu(\mathcal{B})ds &\leq&
  \frac{1}{t}\int_{0}^{t_{1}}\mu(\mathcal{B})ds
  +\frac{1}{t} \left(\log\frac{I(t)}{I(t_{1})}+(t-t_{1})\xi
   \right), \\
   &\leq& \frac{1}{t}\int_{0}^{t_{1}}\mu(\mathcal{B})ds
  +\frac{1}{t} \log\frac{I(t)}{I(t_{1})}+\frac{(t-t_{1})}{t}\xi.
\end{eqnarray*}
The boundedness of $I$ implies that
\begin{equation*}
\limsup\limits_{t \to +\infty}\sup\limits_{y_{0}\in K}
\frac{1}{t}\int_{0}^{t}\mu(\mathcal{B})ds \leq \xi<0.
\end{equation*}
Thus,
$$\bar{q}_{2}=\limsup\limits_{t \to +\infty}\sup\limits_{y_{0}\in K}
\frac{1}{t}\int_{0}^{t}\mu(\mathcal{B})ds \leq \xi<0.
$$
This completes the proof that each solution converges to a steady
state. It remains to note that it has already been shown that when
$\mathcal{R}''_0>1$ no solution can converge to a steady state on
the boundary.
\end{proof}

\section{Numerical simulations}\label{sec6}
In this section, we present some numerical simulations to complement
the theoretical results obtained in the previous sections.
\\\indent
\begin{figure}[!h]
\centering
\begin{subfigure}[b]{0.32\textwidth}
\includegraphics[angle=0,height=3cm,width=\textwidth]{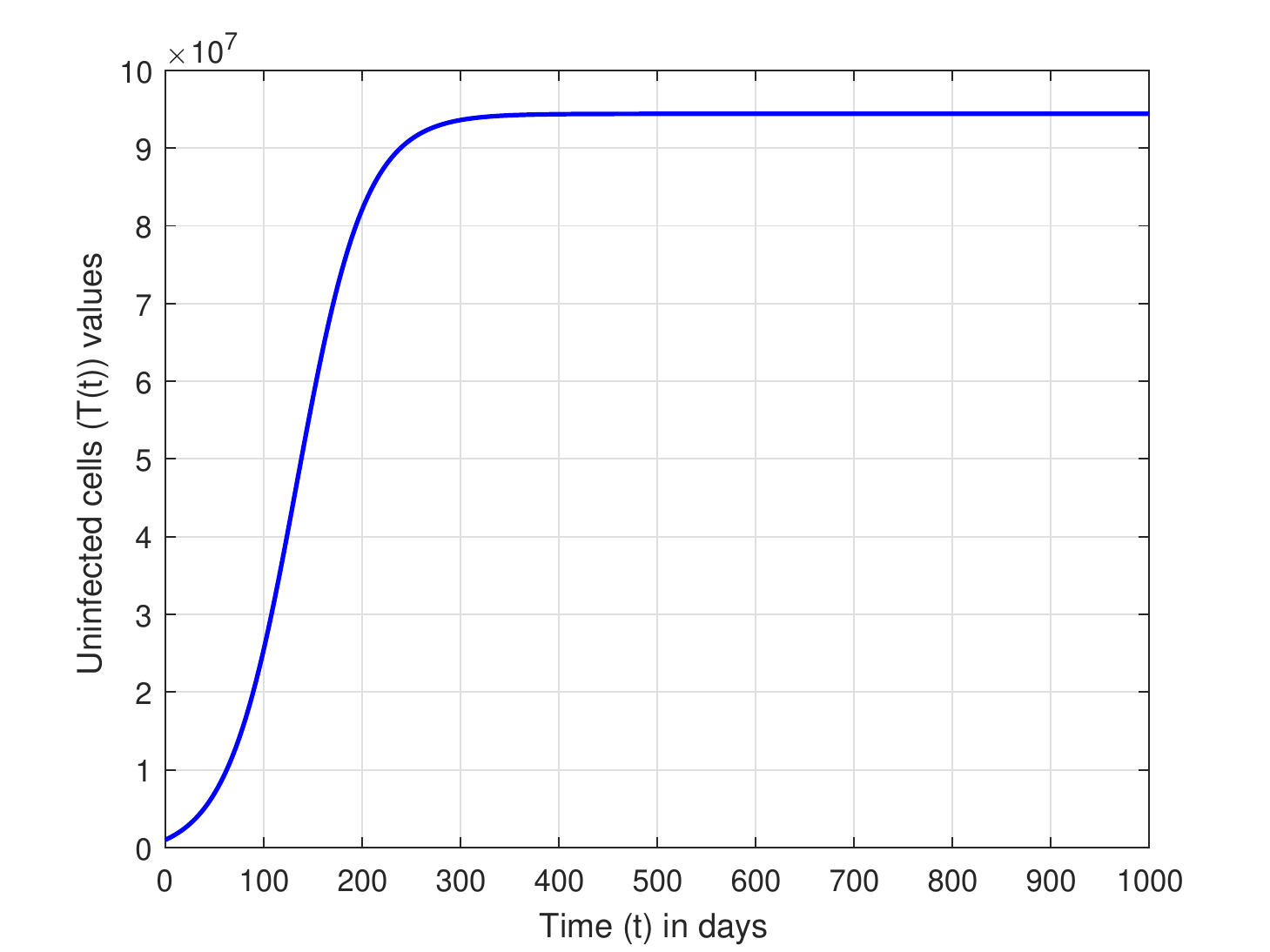}
\end{subfigure}
\begin{subfigure}[b]{0.32\textwidth}
\includegraphics[angle=0,height=3cm,width=\textwidth]{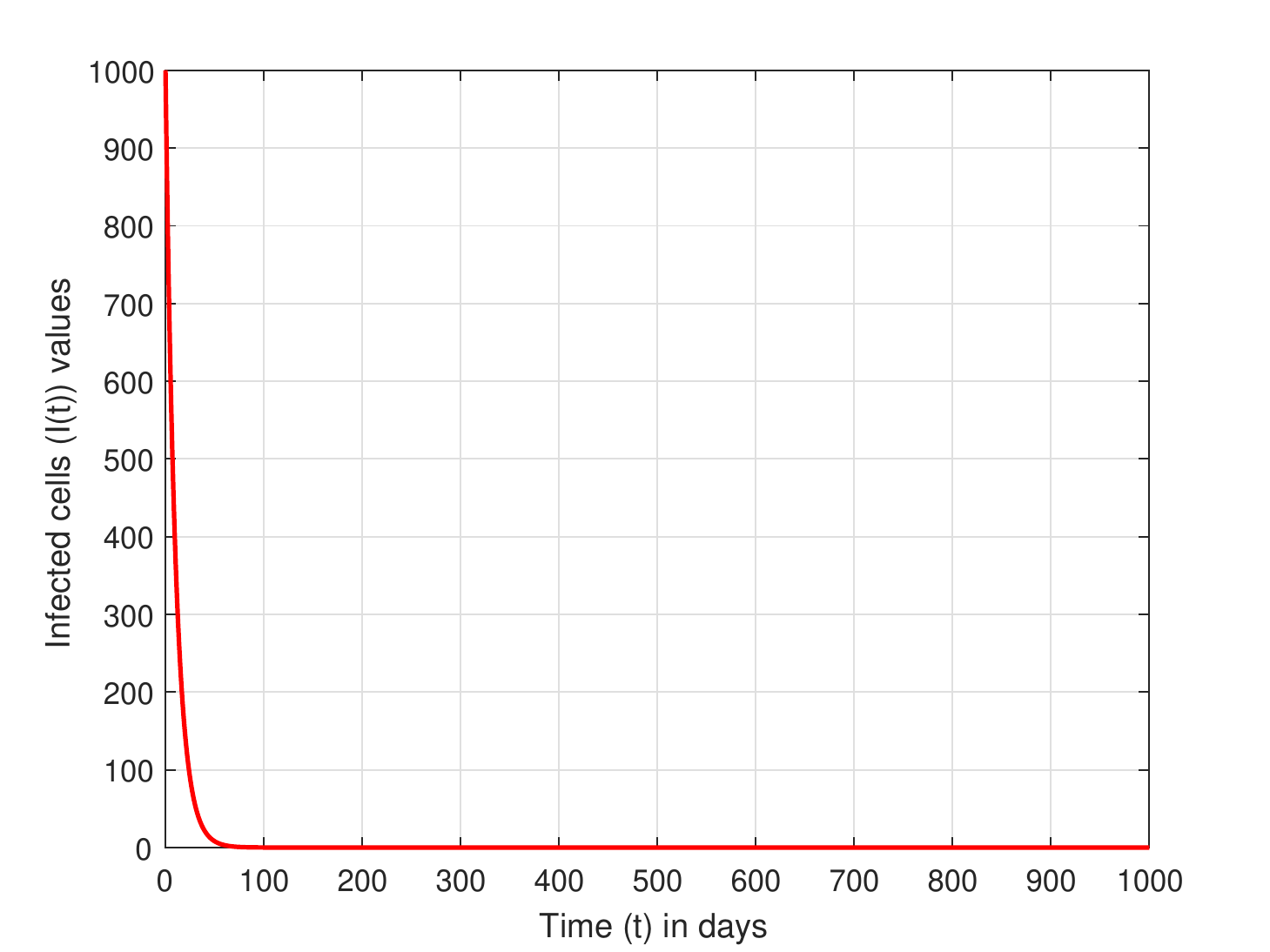}
\end{subfigure}
\begin{subfigure}[b]{0.32\textwidth}
\includegraphics[angle=0,height=3cm,width=\textwidth]{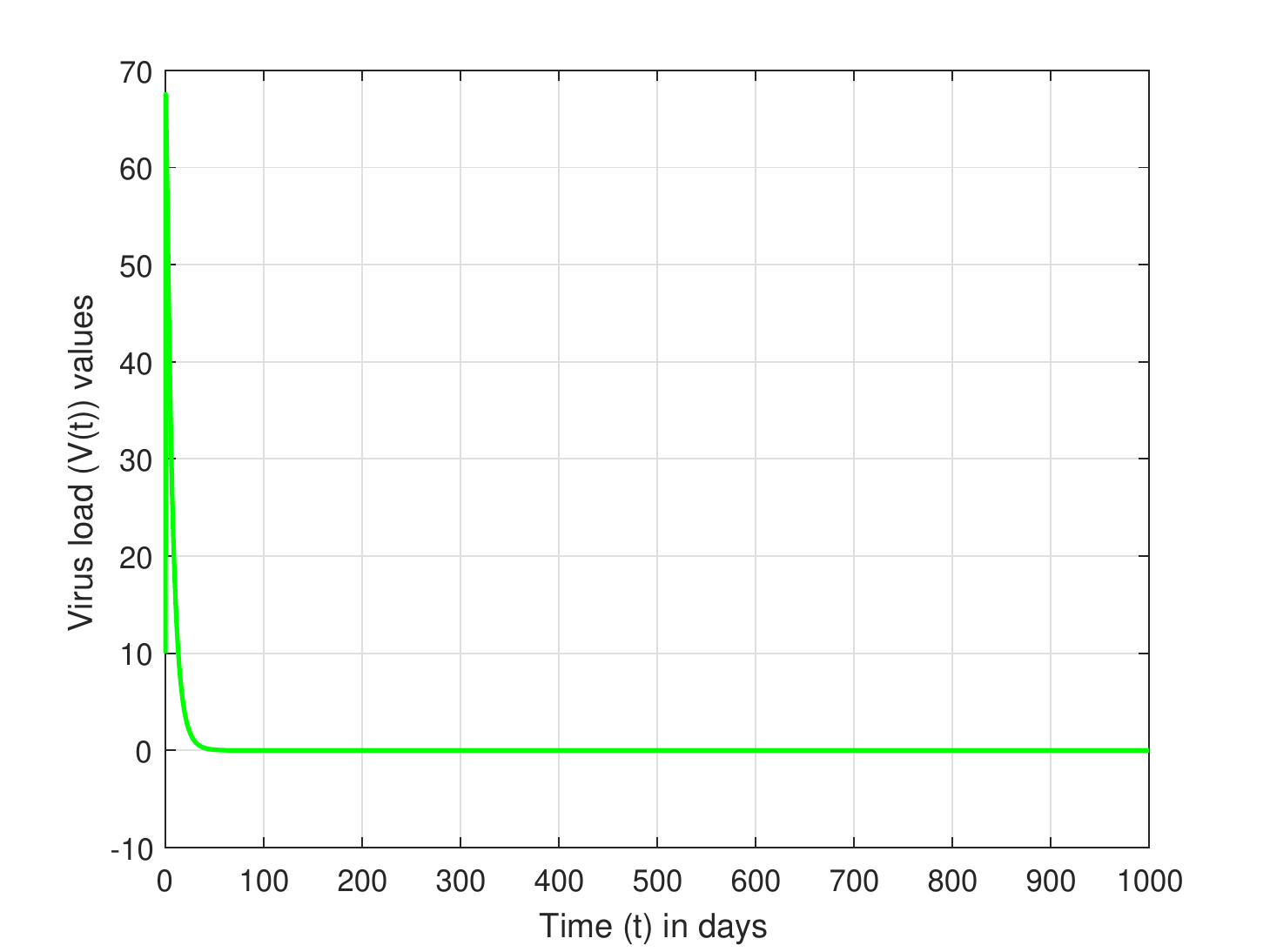}
\end{subfigure}
\begin{subfigure}[b]{0.32\textwidth}
\includegraphics[angle=0,height=3cm,width=\textwidth]{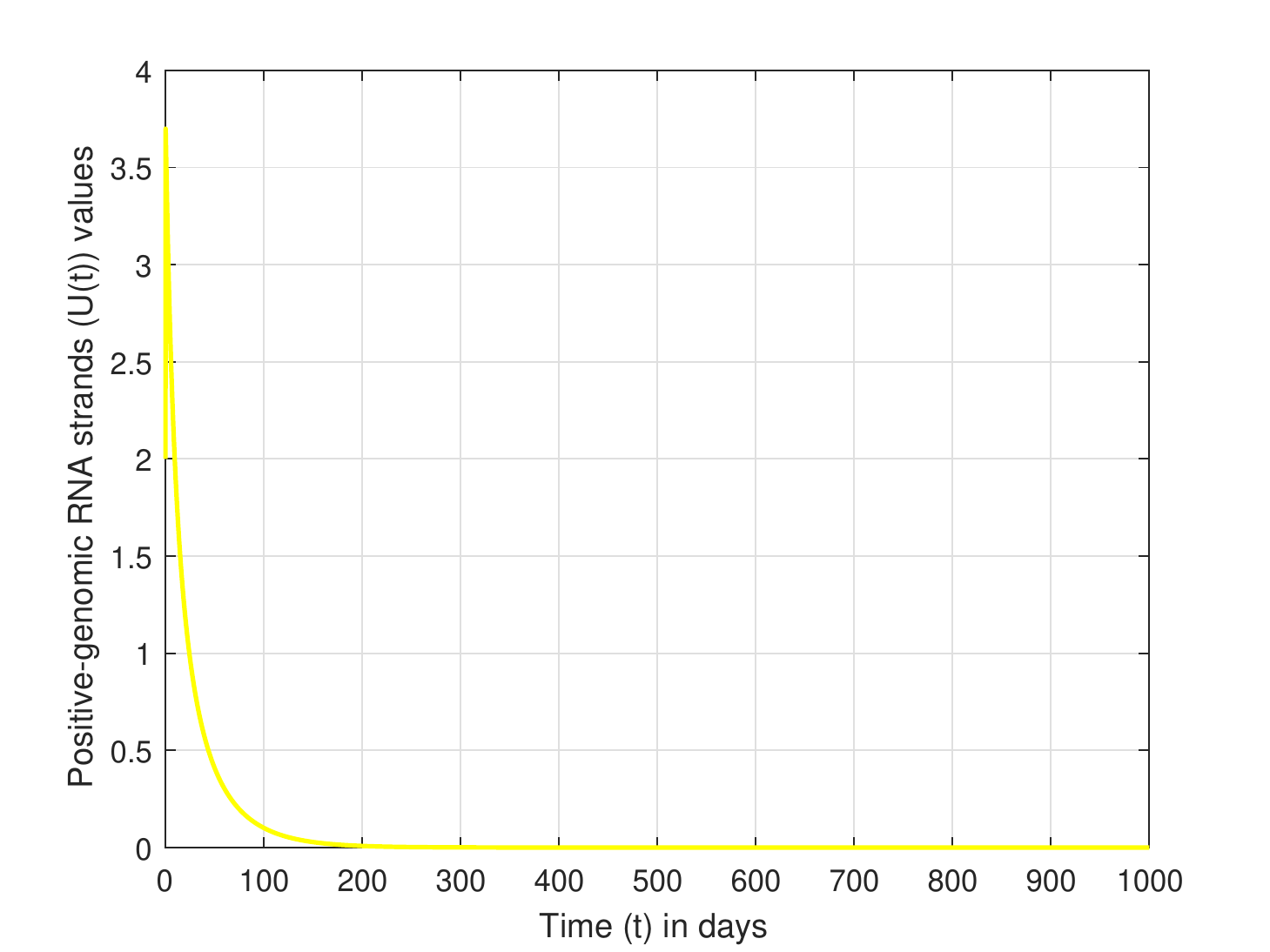}
\end{subfigure}
\begin{subfigure}[b]{0.32\textwidth}
\includegraphics[angle=0,height=3cm,width=\textwidth]{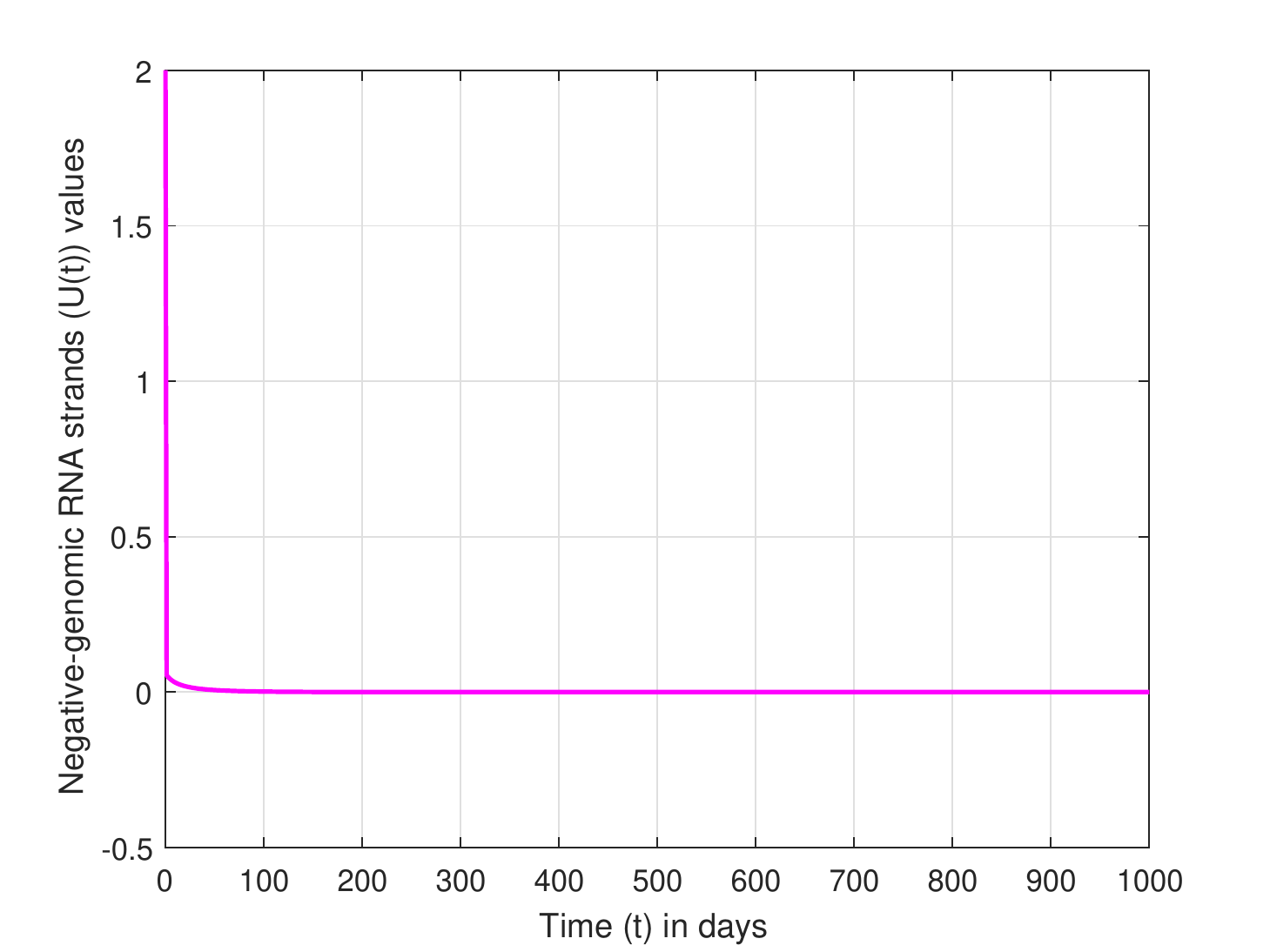}
\end{subfigure}
\begin{subfigure}[b]{0.32\textwidth}
\includegraphics[angle=0,height=3cm,width=\textwidth]{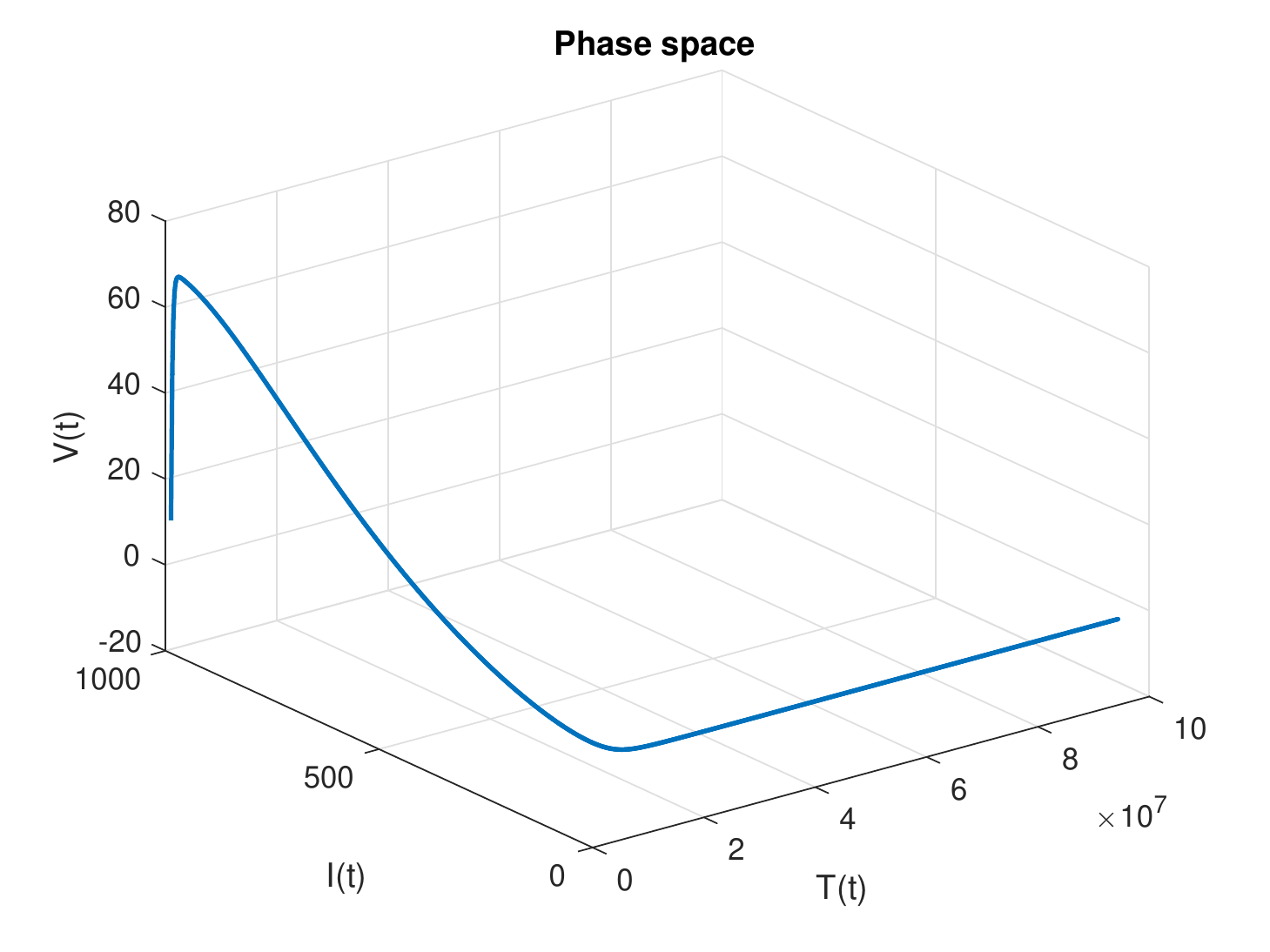}
\end{subfigure}
\begin{subfigure}[b]{0.32\textwidth}
\includegraphics[angle=0,height=3cm,width=\textwidth]{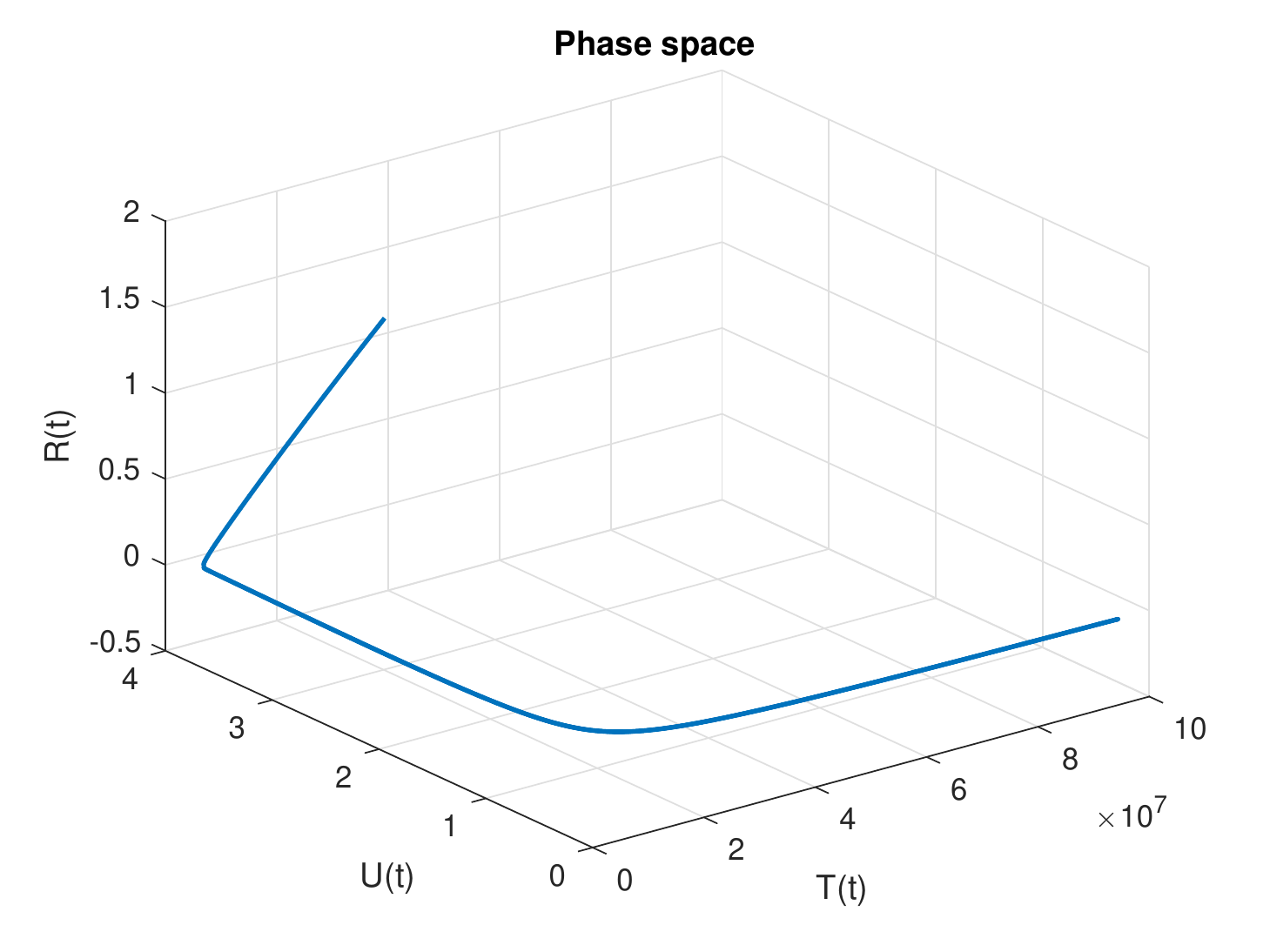}
\end{subfigure}
\caption{Simulations of Initial value problem (\ref{a6}) using
various initial conditions when $s= 3 \times 10^{4}$; $d= 2\times
10^{-3}$; $ r_{T}=3\times 10^{-2}$; $ r_{I}= 10^{-3}$;
$T_\textrm{max}= 10^{8}$; $b = 5.44\times10^{-4}$; $\delta =
0.0987$; $c = 8\times10^{-1}$ ; $\rho = 0.8898$; $ \alpha = 9.5$;
$\beta = 30$; $ \epsilon = 0.95 $; $\sigma = 30$; $U_\textrm{max}=
30$; $\gamma = 0.5$; $E_{0}=(9.439273688\times10^{7}, 0, 0, 0, 0)$
(Such that $\mathcal{R}'_{0}=0.95 $ and
$\mathcal{R}''_{0}=0.0552$).} \label{fig1}
\end{figure}
Figure~\ref{fig1} illustrates the case $\mathcal{R}'_{0} < 1$ and
$\mathcal{R}''_{0} < 1$. From this figure, we observe that the
trajectories converge to the HCV-free equilibrium $E_{0}$. This
corresponds to the case where the equilibrium is globally
asymptotically stable. In this case, the infection could disappear
within the host.
\\\indent
\begin{figure}[!h]
\centering
\begin{subfigure}[b]{0.32\textwidth}
\includegraphics[angle=0,height=3cm,width=\textwidth]{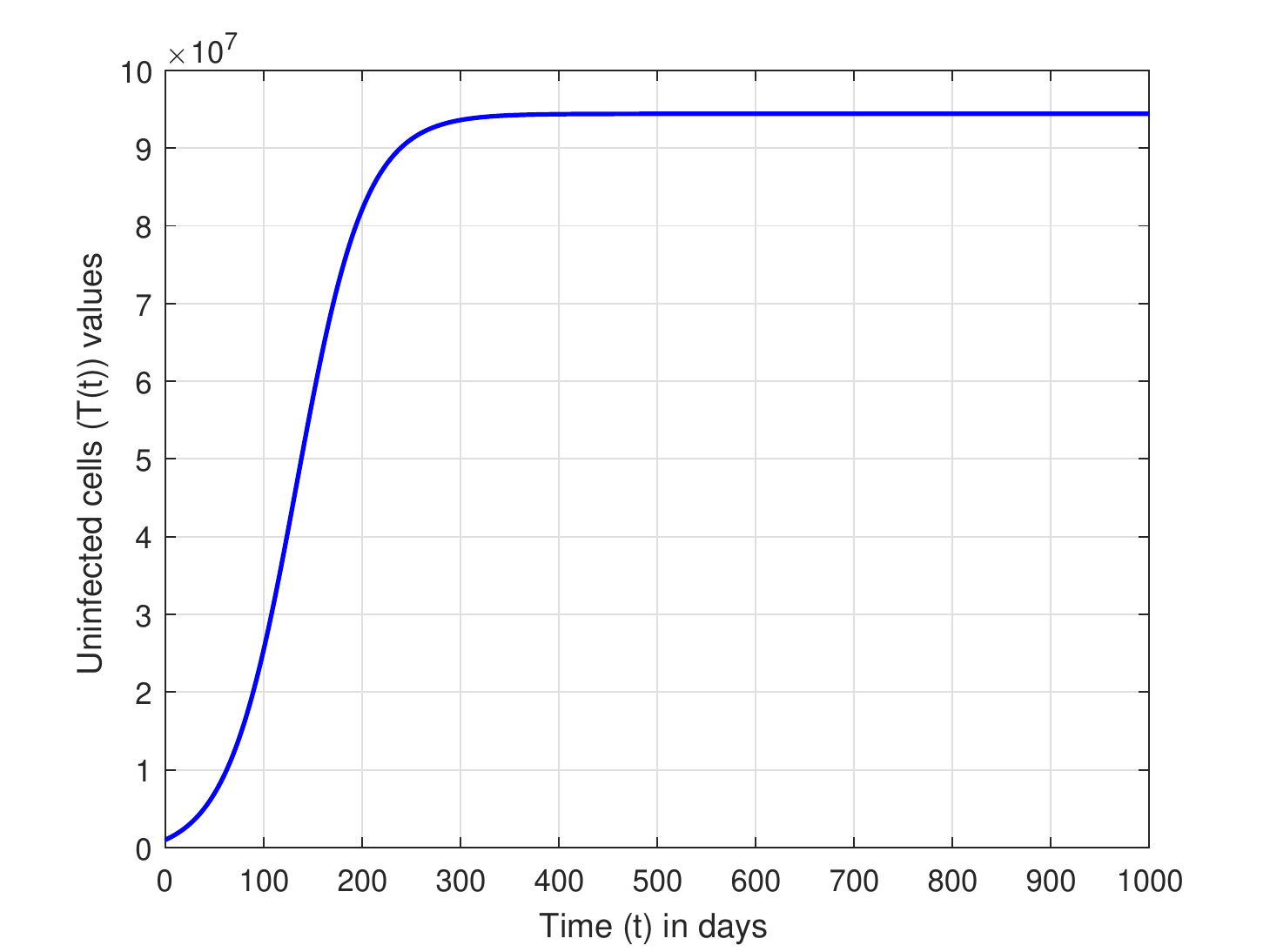}
\end{subfigure}
\begin{subfigure}[b]{0.32\textwidth}
\includegraphics[angle=0,height=3cm,width=\textwidth]{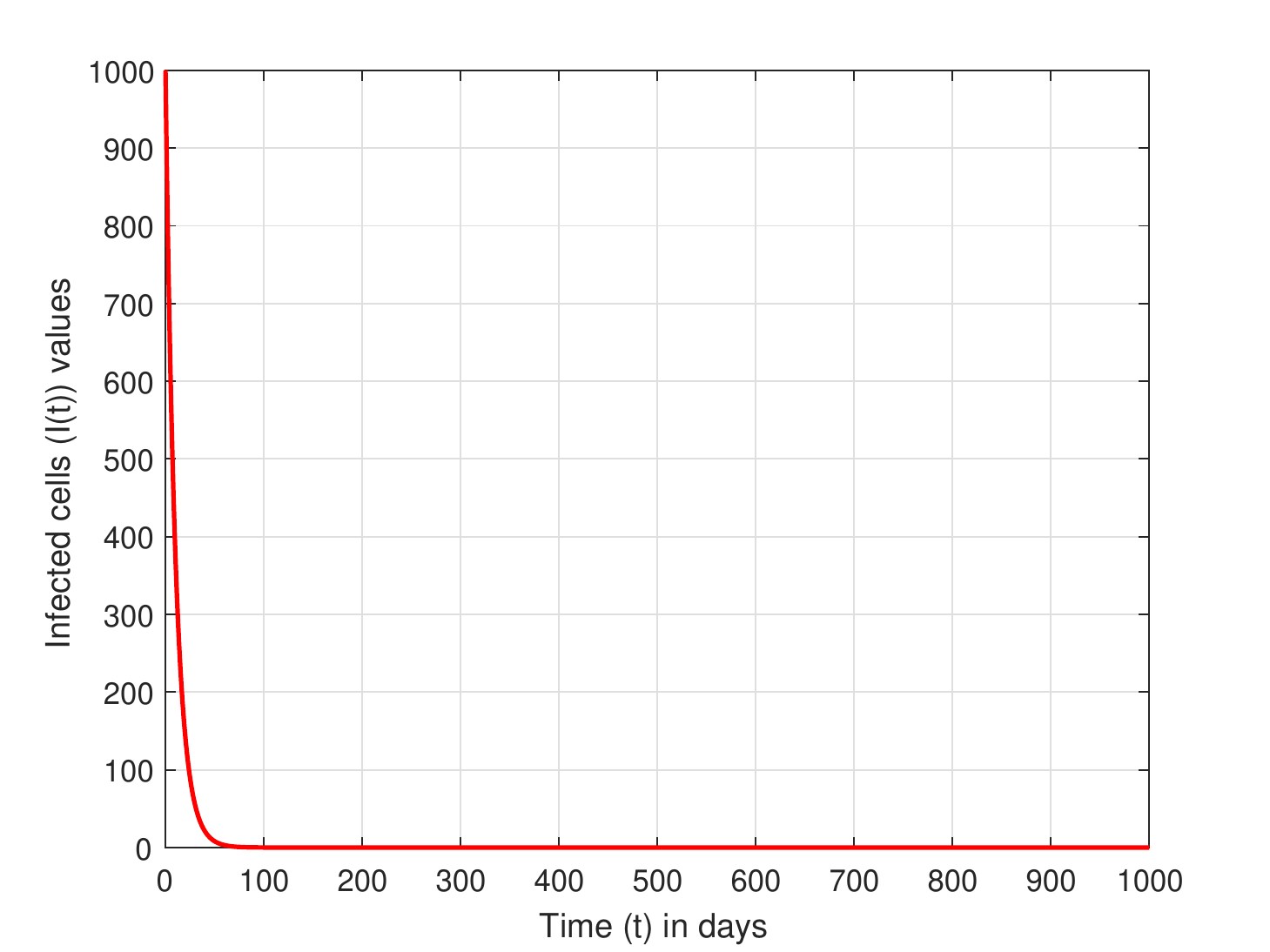}
\end{subfigure}
\begin{subfigure}[b]{0.32\textwidth}
\includegraphics[angle=0,height=3cm,width=\textwidth]{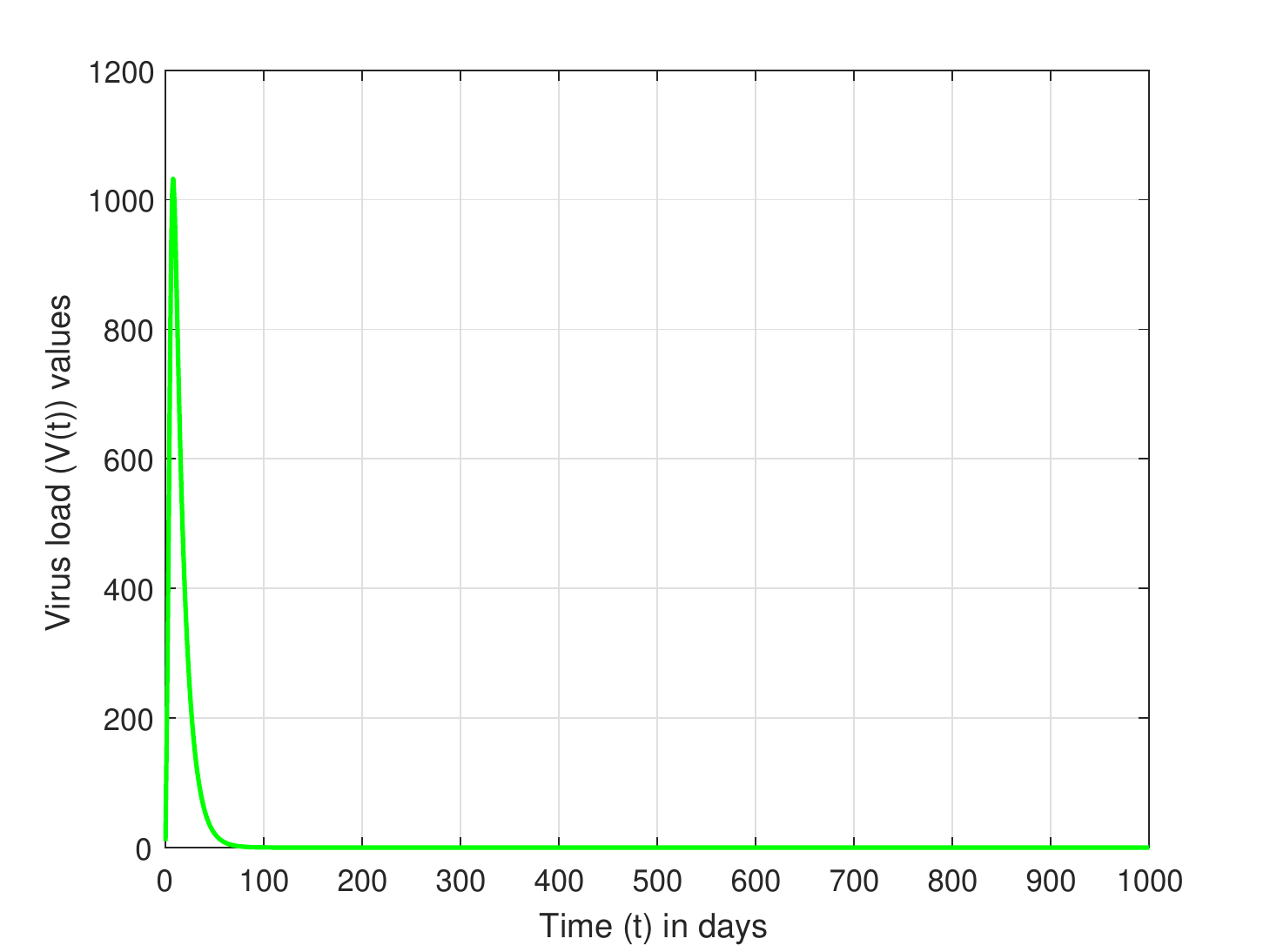}
\end{subfigure}
\begin{subfigure}[b]{0.32\textwidth}
\includegraphics[angle=0,height=3cm,width=\textwidth]{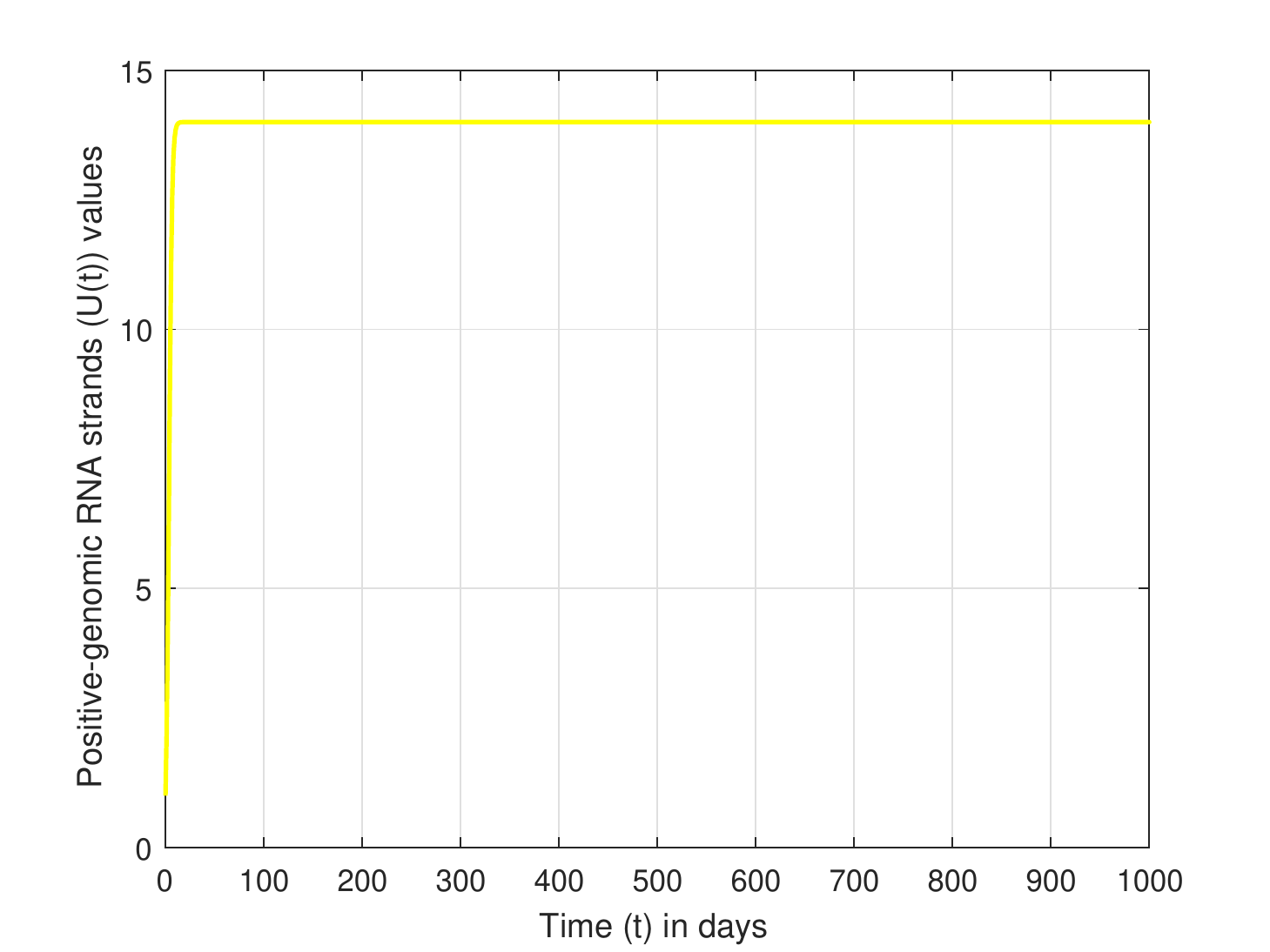}
\end{subfigure}
\begin{subfigure}[b]{0.32\textwidth}
\includegraphics[angle=0,height=3cm,width=\textwidth]{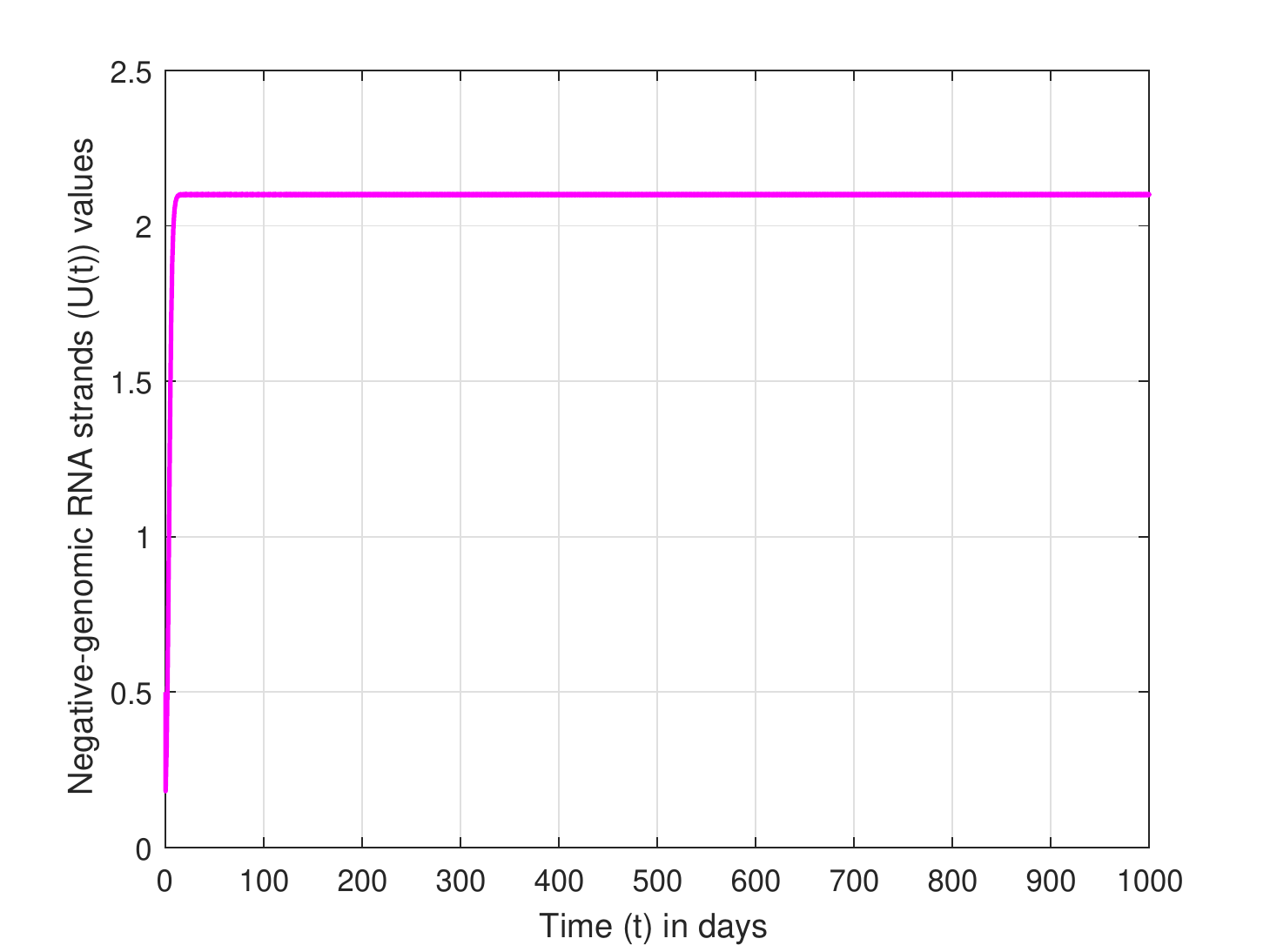}
\end{subfigure}
\begin{subfigure}[b]{0.32\textwidth}
\includegraphics[angle=0,height=3cm,width=\textwidth]{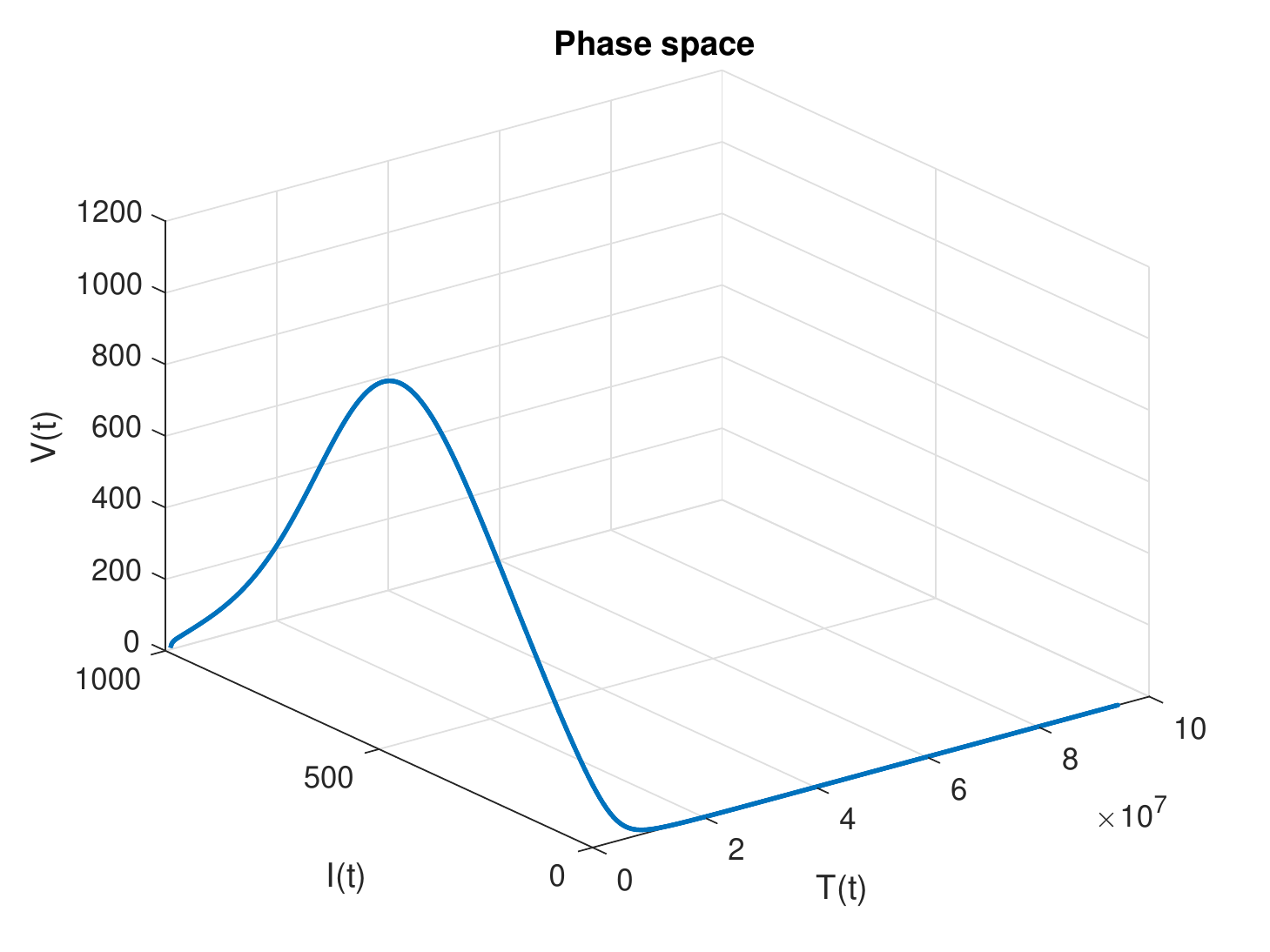}
\end{subfigure}
\begin{subfigure}[b]{0.32\textwidth}
\includegraphics[angle=0,height=3cm,width=\textwidth]{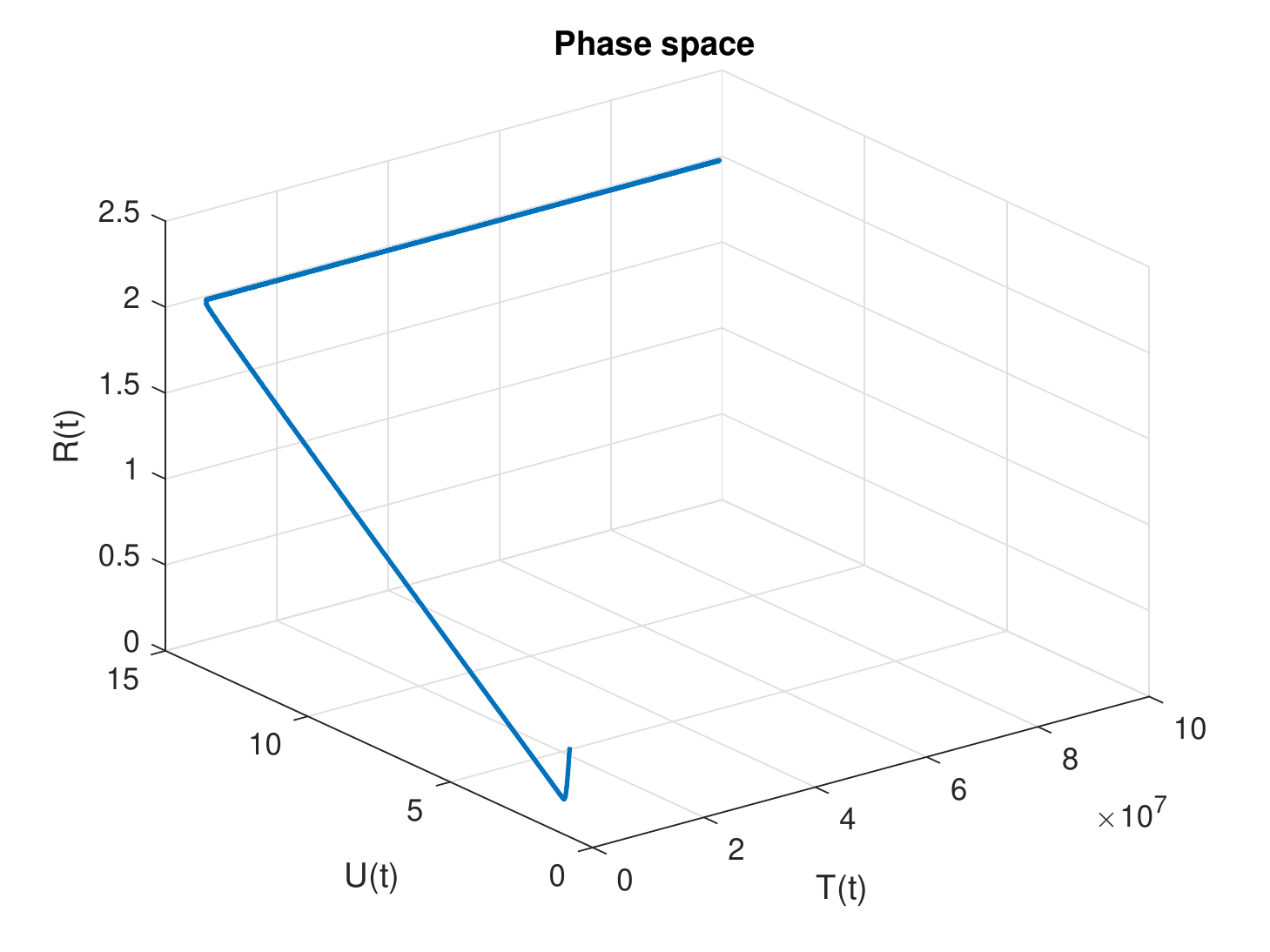}
\end{subfigure}
\caption{Simulations of Initial value problem (\ref{a6}) using
various initial conditions when $s= 3 \times 10^{4}$; $d= 2\times
10^{-3}$; $ r_{T}=3\times 10^{-2}$; $ r_{I}= 10^{-3}$;
$T_\textrm{max}= 10^{8}$; $b = 5.44\times10^{-4}$; $\delta =
0.0987$; $c = 8\times10^{-1}$ ; $\rho = 0.8898$; $ \alpha = 5$;
$\beta = 10$; $ \epsilon = 0.1 $; $\sigma = 30$; $U_\textrm{max} =
30$; $\gamma=0.8$; $E'_{0}=(9.439273688\times10^7, 0, 0, 4.6, 2.20)$
(Such that $\mathcal{R}''_{0}=0.0159$ and $\mathcal{R}'_{0}=
1.8750$).} \label{fig12}
\end{figure}
 Figure~\ref{fig12} illustrates the case
$\mathcal{R}''_{0}< 1$ and $\mathcal{R}'_{0}> 1$. We observe that
the trajectories converge to the second HCV-free equilibrium
$E'_{0}$. This corresponds to the case where $E'_{0}$ is globally
asymptotically stable. In this case, the infection could disappear
within the host but the viral replication units will persist.
\\\indent
\begin{figure}[!h]
\centering
\begin{subfigure}[b]{0.32\textwidth}
\includegraphics[angle=0,height=3cm,width=\textwidth]{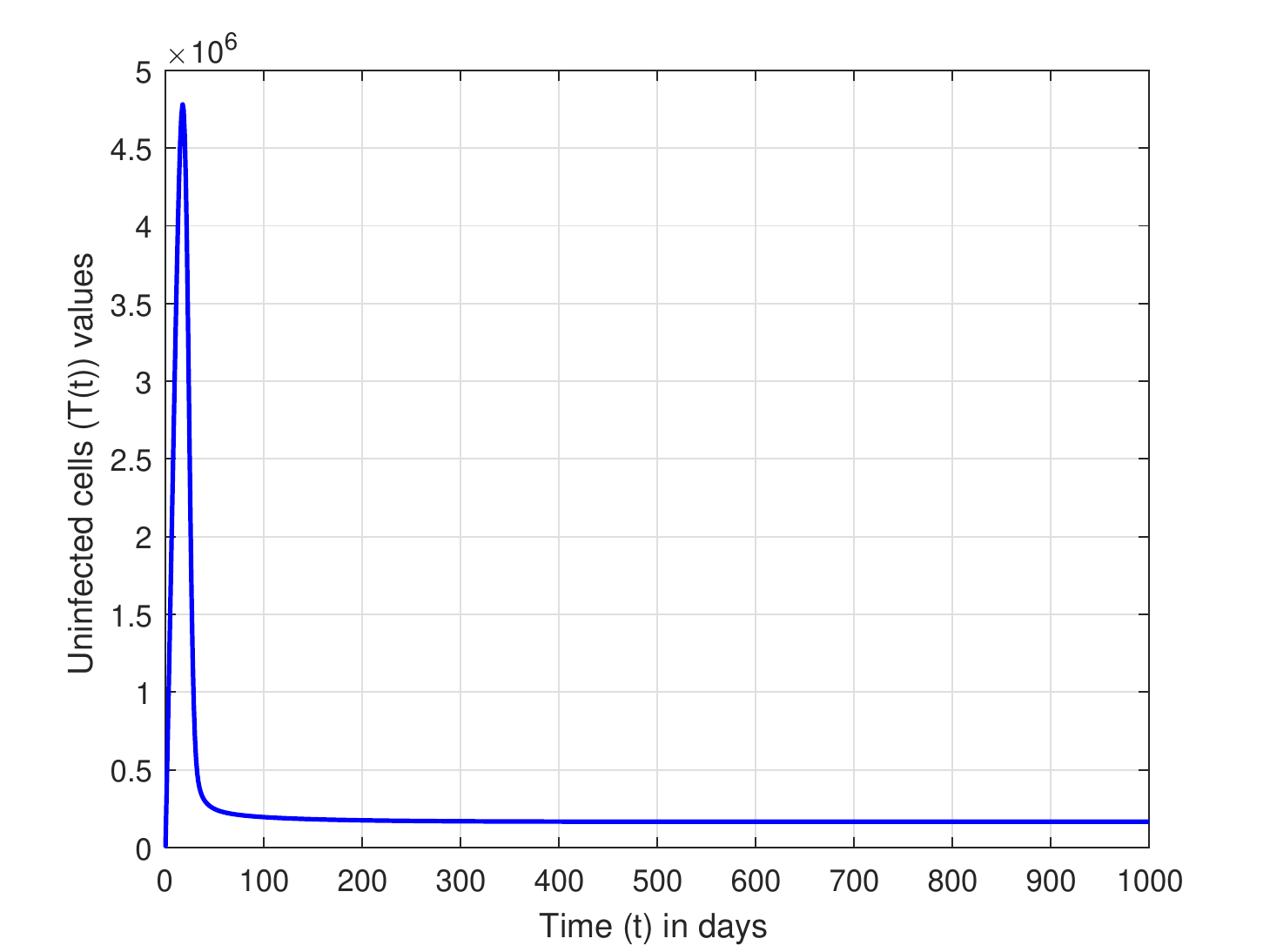}
\end{subfigure}
\begin{subfigure}[b]{0.32\textwidth}
\includegraphics[angle=0,height=3cm,width=\textwidth]{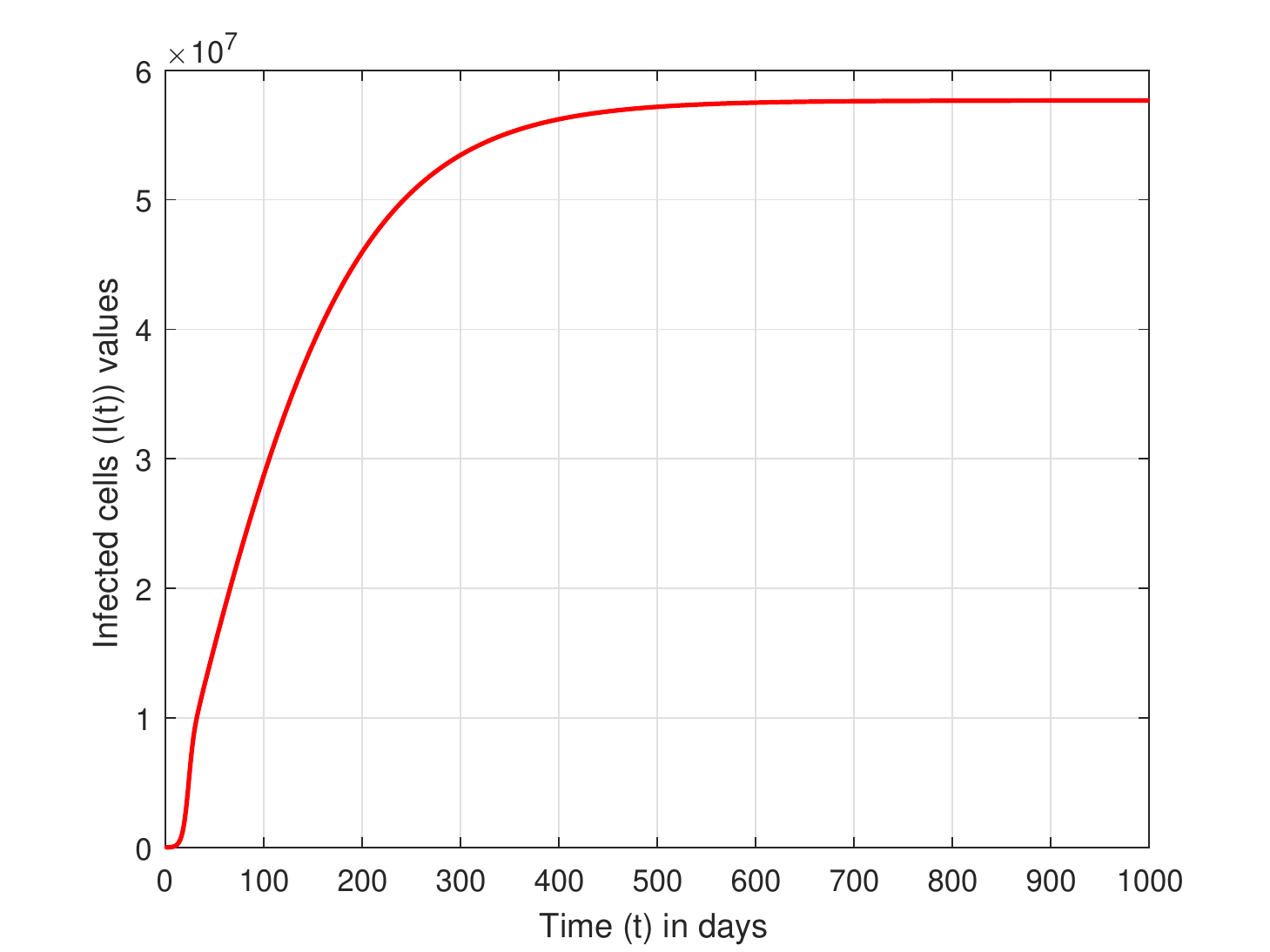}
\end{subfigure}
\begin{subfigure}[b]{0.32\textwidth}
\includegraphics[angle=0,height=3cm,width=\textwidth]{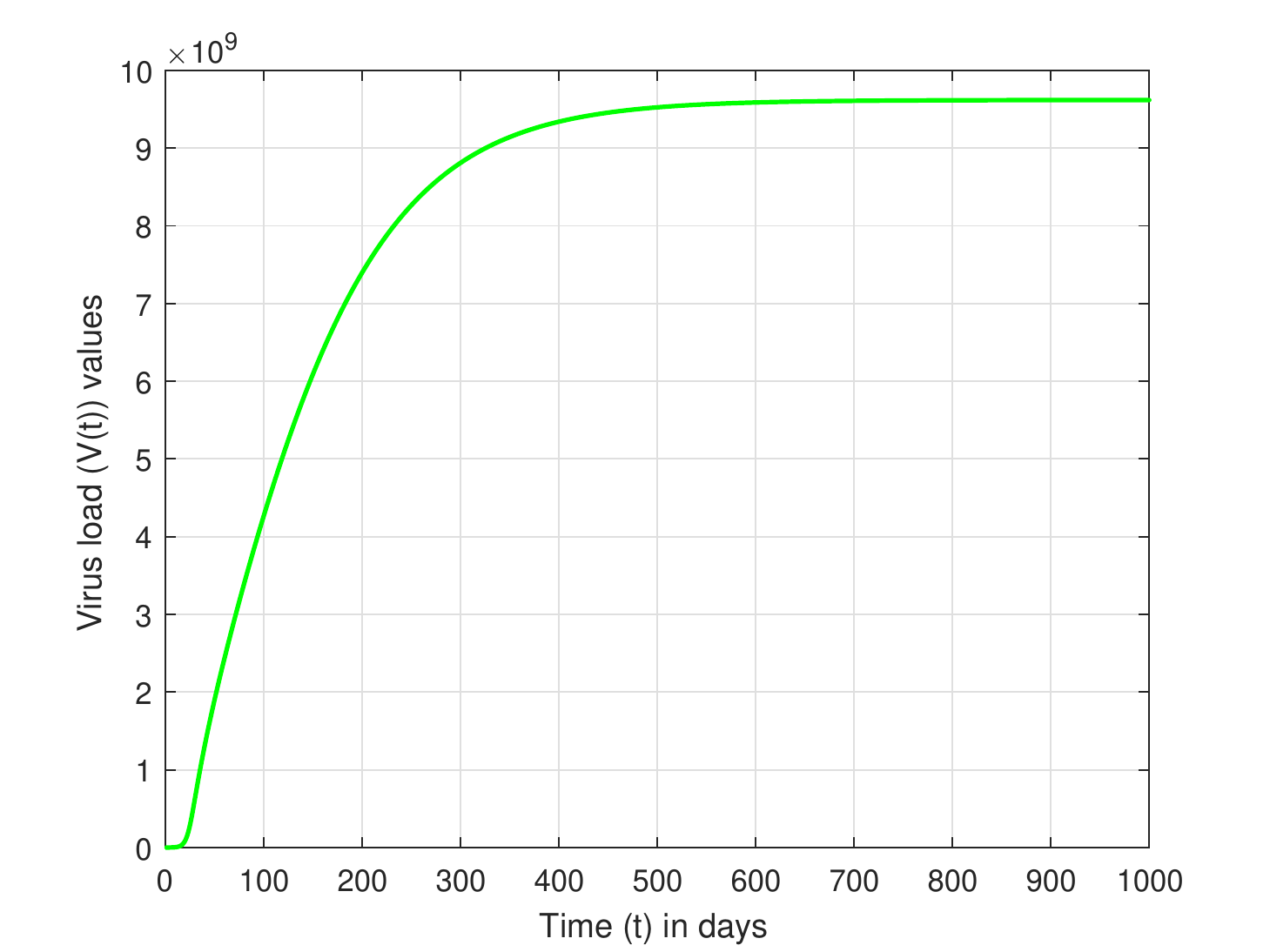}
\end{subfigure}
\begin{subfigure}[b]{0.32\textwidth}
\includegraphics[angle=0,height=3cm,width=\textwidth]{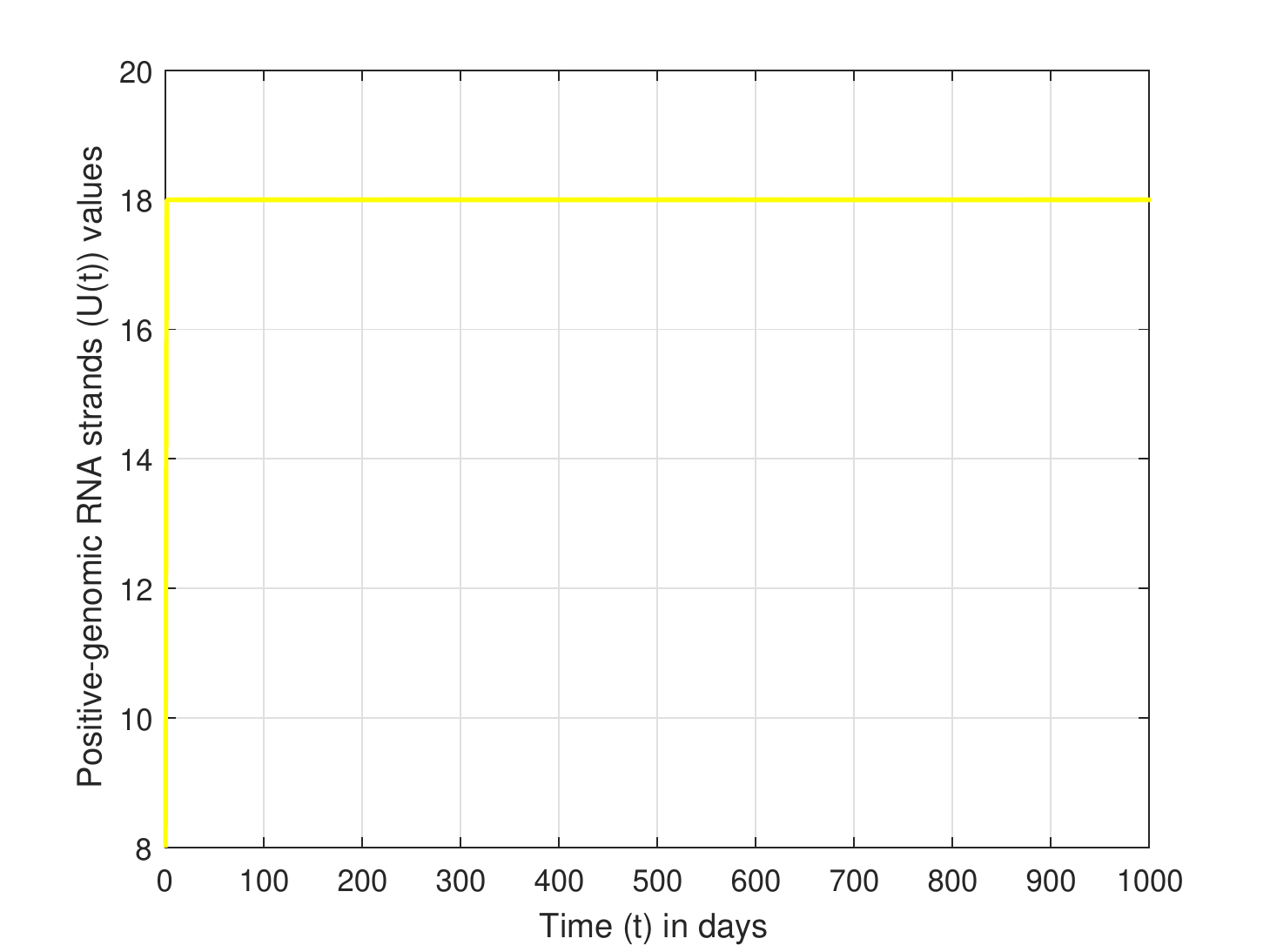}
\end{subfigure}
\begin{subfigure}[b]{0.32\textwidth}
\includegraphics[angle=0,height=3cm,width=\textwidth]{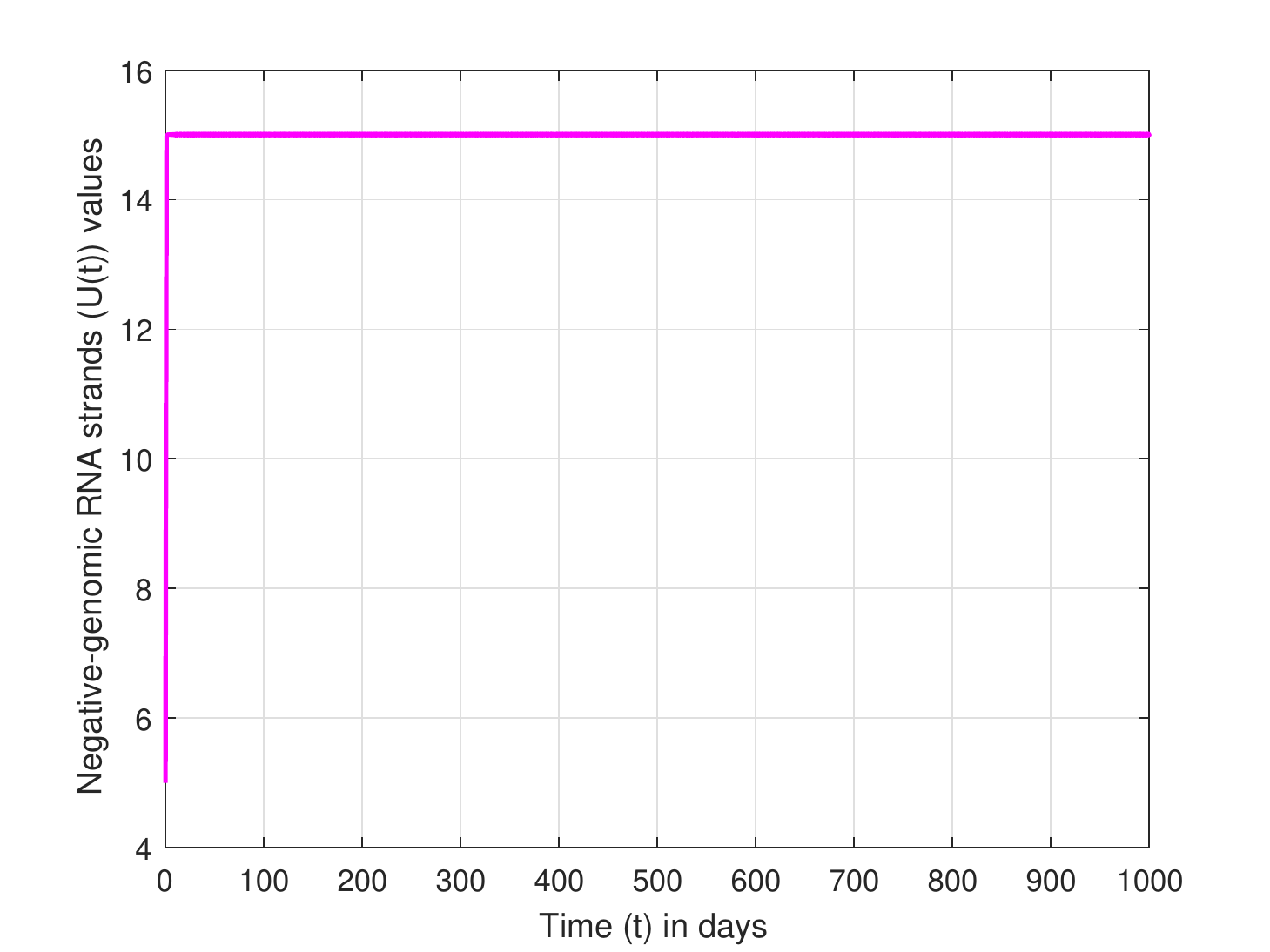}
\end{subfigure}
\begin{subfigure}[b]{0.32\textwidth}
\includegraphics[angle=0,height=3cm,width=\textwidth]{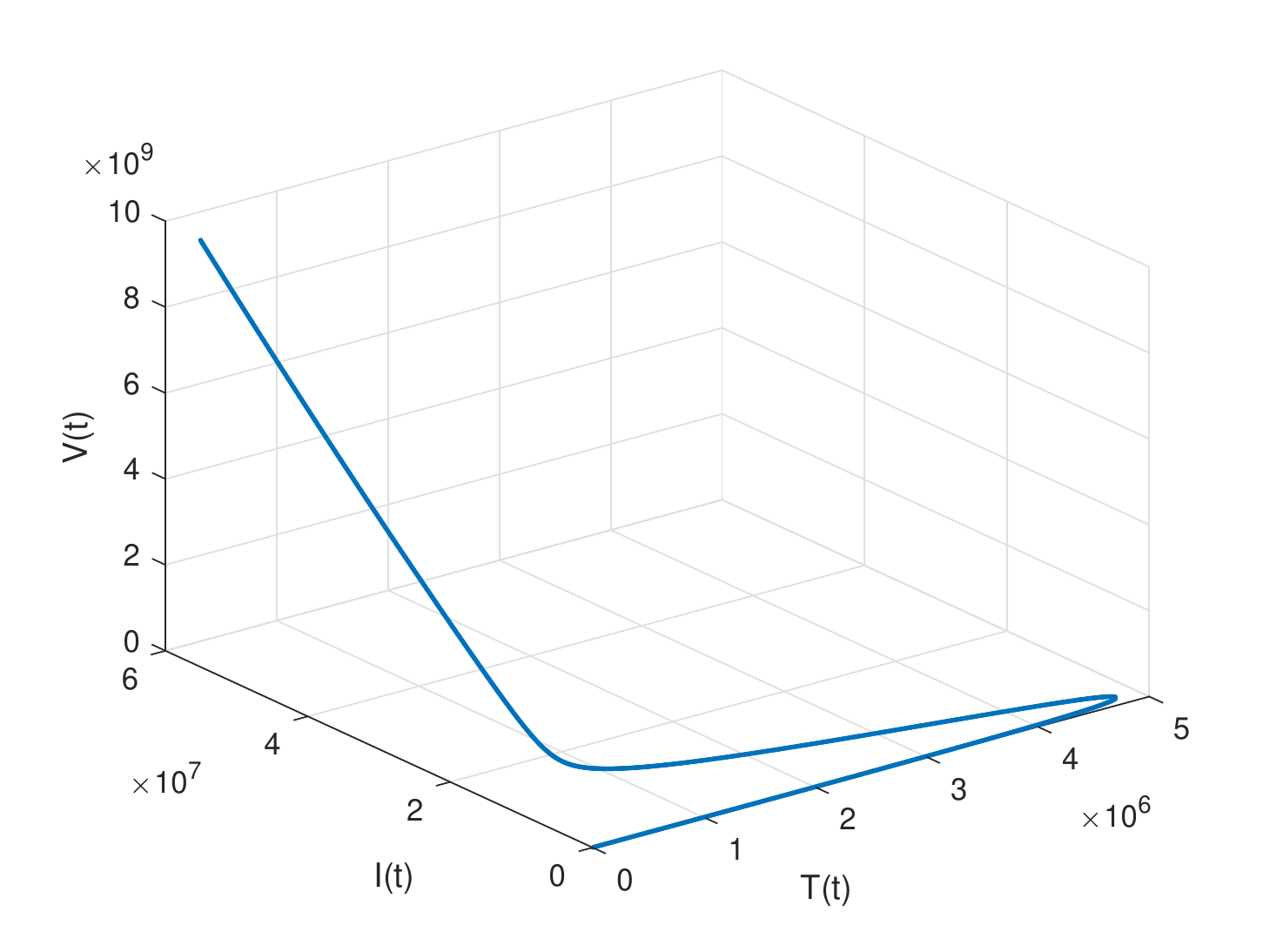}
\end{subfigure}
\begin{subfigure}[b]{0.32\textwidth}
\includegraphics[angle=0,height=3cm,width=\textwidth]{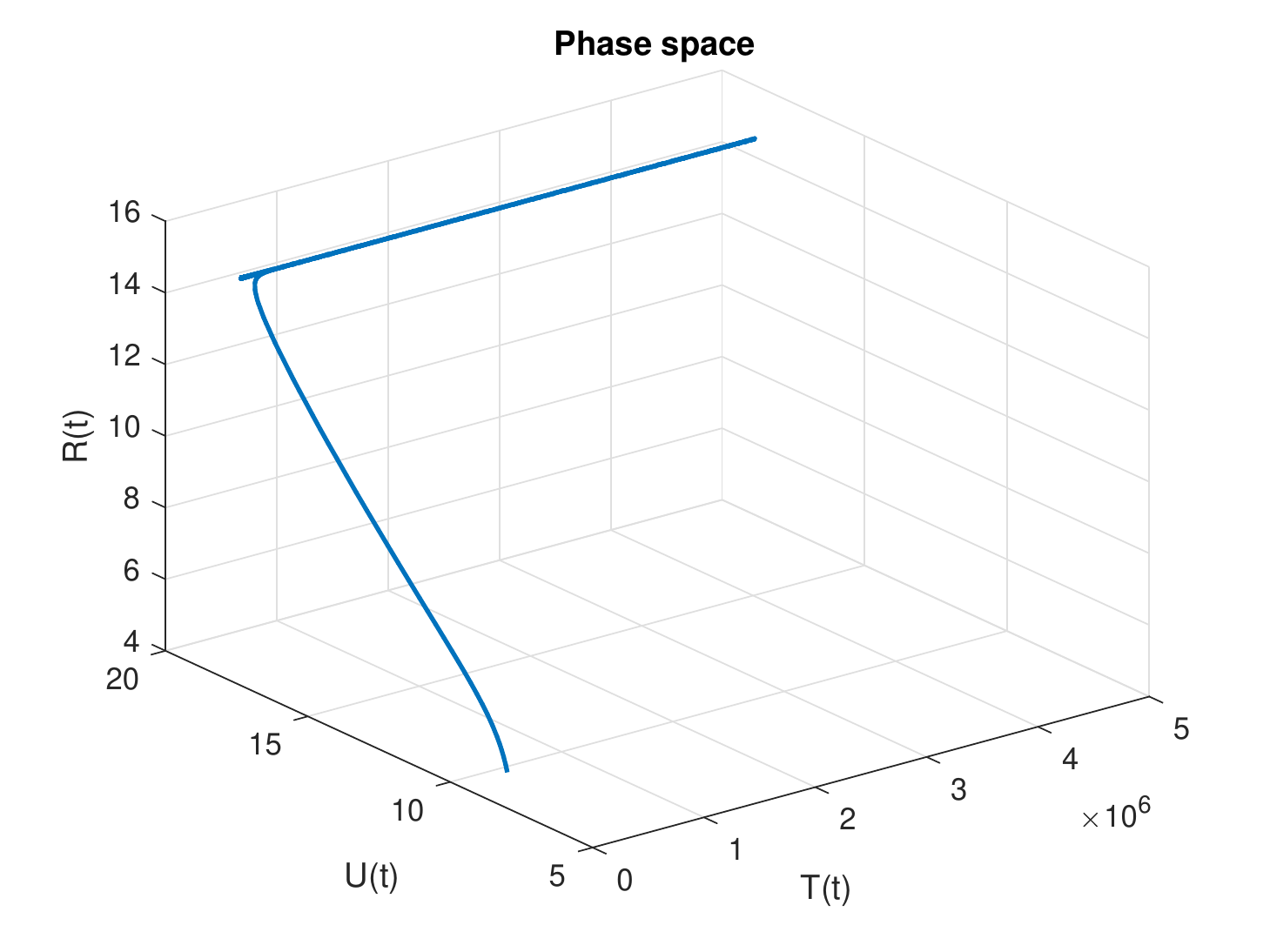}
\end{subfigure}
\caption{Simulations of Initial value problem (\ref{a6}) using
various initial conditions when $s= 3 \times 10^{5}$; $d=
2\times10^{-3}$; $r_{T}= 3\times10^{-2}$; $ r_{I}= 10^{-3}$;
$T_\textrm{max}= 10^{8}$; $b = 0.01$; $ \delta = 10^{-6}$; $c =
8\times10^{-2}$ ; $ \rho = 0.8898$ ; $\alpha= 50$ ; $\beta = 15$;
$\epsilon = 0.5$; $\sigma = 30$ ; $ U_\textrm{max} = 30$ ; $\gamma =
5$; $E^{*}= (0.25\times10^{6}, 5.766414063\times10^{7},
9.617092746\times10^{9}, 18, 15)$ (Such that $\mathcal{R}'_{0}
=6.9235$ and $\mathcal{R}''_{0} =2.5$).} \label{fig2}
\end{figure}
Figure~\ref{fig2} illustrates the case corresponding to
$\mathcal{R}'_{0} > 1$ and $\mathcal{R}''_{0} > 1$. From this
figure, it is seen that the trajectories converge to an infected
equilibrium $E^{*}$. In this case, the infection  persists within
the host.
\\\indent
Apart from, numerical solutions of ODE model system (\ref{a6}), we
also complete the numerical simulations by phase portrait in TIV
space and TUR space.

\section{Conclusion}\label{sec7}
In order to better understand the dynamics of HCV viral infection,
this paper presents a mathematical study on the global dynamics of
improved HCV models based on system~(3) in \cite{guedj2010}. In this
work, we have studied the models describing the dynamics of the
hepatitis C viral cellular and intracellular infection model with
logistic cellular growth. The model includes five equations
illustrating the interaction between the uninfected cells, infected
cells, HCV virus, positive genomic RNA strands and negative strands.
The global existence, the positivity and the boundedness of
solutions are established. The existence of an infected steady state
is also established for certain values of the parameters.
Furthermore, we have studied the local stability of both uninfected
equilibrium and infected equilibrium. Concerning global asymptotic
stability, that of an uninfected equilibrium point was established
by the construction of a suitable Lyapunov function. It was also
shown using the Li-Muldowney global-stability criterion that for
certain values of the parameters every solution converges to a
steady state. Finally, we performed numerical simulations to
illustrate the theoretical results obtained. It would be interesting
to incorporate time delay or spatial dependence into the current
model. These two challenges will be the concerns of future
investigation.

A number of the conclusions of the paper required making
restrictions on the parameters of the model. It would be desirable
to investigate what happens when these restrictions are removed:
which of the conclusions extend? It would also be desirable to
understand the biological meaning of these restrictions. Let us just
comment on one of these, the inequality $r_I-\delta>0$. This means,
roughly speaking, that if all the liver cells were infected the
liver would be able to sustain itself. Note that in practise it
could be that during a chronic hepatitis C infection most of the
hepatocytes are infected. Thus it seems intuitively that this
inequality is related to the condition that a chronic infection
persists. Finally it would be desirable to make a broad comparison
of the properties of the model in this paper with those of other
models for hepatitis C or other related diseases such as hepatitis B
in the literature. Note, for instance, that in contrast to what we
found here, in the model for hepatitis B in \cite{hews2010}, which
also uses the standard incidence function, it does sometimes happen
that $T+I\to 0$. This is connected to the fact that while we choose
$s>0$ the model of \cite{hews2010} corresponds to the case $s=0$.
This in turn is related to the question whether the population of
hepatocytes is maintained by cell division in the liver or whether
is also supported by migration of cells from outside.

\begin{thebibliography}{10}
\expandafter\ifx\csname url\endcsname\relax
  \def\url#1{\texttt{#1}}\fi
\expandafter\ifx\csname urlprefix\endcsname\relax\def\urlprefix{URL
}\fi \expandafter\ifx\csname href\endcsname\relax
  \def\href#1#2{#2} \def\path#1{#1}\fi

\bibitem{Shepardetal2005}
C.~W. {Shepard}, L.~{Finelli}, M.~J. {Alter}, Global epidemiology of
hepatitis
  {C} virus infection, Lancet Infect.Dis. 5~(9) (2005) 558--567.

\bibitem{SeeffandHoofnagle2003}
L.~Seeff, J.~Hoofnagle, Appendix :the national institutes of health
consensus
  development conference management of hepatitis {C} 2002, Clin.LiverDis 7~(1)
  (2003) 261--287.

\bibitem{who2016}
World {H}ealth {O}rganization, global report on access to hepatitis
{C}
  treatment-{F}ocus on overcoming barriers, Tech. rep., Available online at
  https://www.who.int/hepatitis /publications/hep-c-access-report/en/ (2016).

\bibitem{Anangue22019}
A.~{Nangue}, C.~{Fokoue}, R.~{Poumeni}, The global stability
analysis of a
  mathematical cellular model of hepatitis {C} virus infection with
  non-cytolytic process, Applied Mathematics and Physics 7 (2019) 1531--1546.

\bibitem{Xiang}
L.~{Song}, C.~{Ma}, Q.~{Li}, A.~{Fan}, K.~{Wang}, Global dynamics of
a viral
  infection model with full logistic terms and antivirus treatments, Int. J.
  Biomath 10~(1) (2017) 1750012--1--1750012--24, dOI:
  10.1142/S1793524517500127.

\bibitem{hattaf1}
K.~{Hattaf}, N.~{Yousfi}, Global stability of a virus dynamics model
with cure
  rate and absorption, Journal of the Egyptian Mathematical Society 22 (2014)
  386--389.

\bibitem{Gong}
M.~S.~F. {Chong}, S.~{Masitah}, L.~{Crossley}, A.~{Madzvamuse}, The
stability
  analyses of the mathematical models of hepatitis {C} virus infection, Modern
  Applied Science 9~(3) (2015) 250--271, ISSN 1913-1844.

\bibitem{nangue2}
A.~{Nangue}, T.~{Donfack}, D.~A. {Ndode Yafago}, Global dynamics of
an
  hepatitis {C} virus mathematical cellular model with a logistic term, Eur. J.
  Pure Appl. Math 12~(3) (2019) 944--959, ISSN 1307-5543.

\bibitem{nangue3}
A.~{Nangue}, Global stability analysis of the original cellular
model of
  hepatitis {C} virus infection under therapy, American Journal of Mathematical
  and Computer Modelling 4~(3) (2019) 58--65, ISSN: 2578-8280.

\bibitem{Neumannetal1998}
A.~U. {Neumann}, N.~P. {Lam}, H.~{Dahari}, D.~R. {Gretch}, T.~E.
{Wiley}, T.~J.
  {Layden}, A.~S. {Perelson}, Hepatitis {C} viral dynamics in vivo and the
  antiviral efficacy of interferon-alpha therapy, Science 282 (1998) 103--107.

\bibitem{Herrmann2003}
E.~{Hermann}, J.~H. {Lee}, G.~{Marinos}, M.~{Modi}, S.~{Zeuzem},
Effect of
  ribavirin on hepatitis {C} viral kinetics in patients treated with pegylated
  interferon, Hepatology 37~(6) (2003) 1351--1358.

\bibitem{guedj2010}
J.~{Guedj}, A.~U. {Neumann}, Understanding hepatitis {C} viral
dynamics with
  direct-acting antiviral agents due to the interplay between intracellular
  replication and cellular infection dynamics, J.Theo.Biology~(267) (2010)
  330--340.

\bibitem{hews2010}
S.~{Hews}, S.~{Eikenberry}, J.~D. {Nagy}, Y.~{Kuang}, Rich dynamics
of a
  hepatitis {B} viral infection model with logistic hepatocyte growth, J. Math.
  Biol.~(60) (2010) 573--590, {DOI} 10.1007/s00285-009-0278-3.

\bibitem{Daharietal2007}
H.~{Dahari}, R.~{Ribeiro}, C.~{Rice}, A.~{Perelson}, Mathematical
modeling of
  subgenomic hepatitis {C} virus replication
  in {H}uh-7 cells, J.Virol. 81~(2).

\bibitem{Eggeretal2002}
D.~{Egger}, B.~{Wolk}, R.~{Gosert}, L.~{Bianchi}, H.~{Blum},
D.~{Moradpour},
  K.~{Bienz}, Expression of hepatitis {C} virus proteins induces distinct
  membrane alterations including a candidate viral replication complex,
  J.Virol. 76~(12) (2002) 5974--5984.

\bibitem{Gosertetal2003}
R.~{Gosert}, D.~{Egger}, V.~{Lohmann}, R.~{Bartenschlager},
H.~{Blum}, K.~B.~D.
  {Moradpour}, Identification of the hepatitis {C} virus {RNA} replication
  complex in {H}uh-7 cells harboring subgenomic replicons, J. Virol 77~(9)
  (2003) 5487--5492.

\bibitem{Daharietal2009}
H.~{Dahari}, B.~{Sainz}, A.~{Perelson}, S.~{Uprichard}, Modeling
subgenomic
  {HCV} {RNA} kinetics during interferon?alpha treatment, J.Virol. 83~(13)
  (2009) 6384--6393.

\bibitem{Ramratnametal1999}
B.~{Ramratnam}, S.~{Bonhoeffer}, J.~{Binley}, A.~{Hurley},
L.~{Zhang}, J.~E.
  {Mittler}, M.~{Markowitz}, J.~P. {Moore}, A.~{Perelson}, D.~{Ho}, Rapid
  production and clearance of {HIV}-1 and hepatitis {C} virus assessed by large
  volume plasma apheresis, Lancet 354 (1999) 1782--1786.

\bibitem{Diek}
O.~{Diekmann}, J.~A.~P. {Heesterbeek}, J.~A.~J. {Metz}, On the
definition and
  the computation of the basic reproduction ratio ${R}_{0}$ in models for
  infectious diseases in heterogeneous populations, J. Math. Biol 28 (1990)
  365--382.

\bibitem{Dietz}
K.~{Dietz}, Density dependence in parasite transmission dynamics,
Parasit,
  Today 4 (1988) 91--97.

\bibitem{vander}
P.~{van den Driessche}, J.~{Watmough}, Reproduction numbers and
sub-threshold
  endemic equilibria for compartmental models of disease transmission,
  Mathematical Biosciences 180 (2002) 29--48.

\bibitem{12}
R.~M. {Anderson}, R.~M. {May}, Infectious Diseases of Humans:
Dynamics and
  Control, Oxford University Press, 1991.

\bibitem{khalil}
H.~{Khalil}, Nonlinear Systems, 3rd Edition, Prentice Hall, New
York, 2002.

\bibitem{smith10}
H.~L. {Smith}, Monotone dynamical systems, {AMS}, {P}rovidence,
1995.

\bibitem{liMuldowney1996}
M.~Y. {Li}, J.~S. {Muldowney}, A geometric approach to the
global-stability
  problems, SIAM J. Math. Anal. 27~(1070-1083).

\bibitem{freedman1994}
H.~{Freedman}, S.~{Ruan}, {Tang}, Uniform persistence and flows near
a closed
  positively invariant set, J. Dyn. Differ. Equ. 6 (1994) 583--600.

\bibitem{MatinJr1974}
R.~H. {Martin Jr.}, Logarithmic norms and projections applied to
linear
  differential systems, J. Math. Anal. Appl. 45 (1974) 432--454.

\end{thebibliography}
\bibliographystyle{amsplain}

\end{document}